\documentclass[11pt]{article}

\RequirePackage{silence}
\usepackage{fullpage}

\usepackage{hdb_macros}

\usepackage[shortlabels]{enumitem}
\usepackage{footnote}
\usepackage{tikz}
\usetikzlibrary{positioning}
\usepackage{pgfplots}
\usepackage{pgfplotstable}
\pgfplotsset{compat=1.18}
\usepackage[linesnumbered,ruled,noend]{algorithm2e}

\makeatletter 
\g@addto@macro{\@algocf@init}{\SetKwInOut{KwSubroutine}{Subroutine}} 
\makeatother

\theoremstyle{plain}

\crefname{maintheorem}{Theorem}{Theorems}

\RequirePackage[font=small,labelfont=bf]{caption}

\newcommand{\dualMMexp}{\ensuremath{\omega^{\perp}}}
\newcommand{\cR}{{\mathcal{R}}}
\newcommand{\cM}{{\mathcal{M}}}
\newcommand{\cN}{{\mathcal{N}}}

\newcommand{\MM}{\problem{MM}}

\newcommand{\MMV}{\problem{MMV}}
\newcommand{\ColMMV}{\problem{Column\textrm{-}Wise\textrm{-}MMV}}

\newcommand{\FSMM}{\problem{FSMM}}
\newcommand{\OSMM}{\problem{OSMM}}

\newcommand{\ct}{\tilde{c}}

\let\oldnl\nl
\newcommand{\nonl}{\renewcommand{\nl}{\let\nl\oldnl}}

\newcommand{\crefpart}[2]{%
	\hyperref[#2]{\namecref{#1}~\labelcref*{#1}~\ref*{#2}}%
}

\crefname{algocf}{Algorithm}{Algorithms}
\Crefname{algocf}{Algorithm}{Algorithms}

\newcommand{\din}{{\delta_{\textrm{in}}}}
\newcommand{\dout}{{\delta_{\textrm{out}}}}

\title{Output-Sparse Matrix Multiplication Using Compressed Sensing}

\author{
	Huck Bennett\thanks{University of Colorado Boulder. Email: \email{huckbennett@gmail.com}. Supported by NSF Grant CCF-2312297.}
	\and
	Karthik Gajulapalli\thanks{Georgetown University. Email: \email{kg816@georgetown.edu}. Supported by NSF grant CCF-2338730.}
	\and
	Alexander Golovnev\thanks{Georgetown University. Email: \email{alexgolovnev@gmail.com}. Supported by NSF grant CCF-2338730.}
	\and
	Evelyn Warton\thanks{Oregon State University. Email: \email{wartone@oregonstate.edu}.}
}

\date{\today}
\begin{document}
	\maketitle
	
	\begin{abstract}
		We give two algorithms for \emph{output-sparse matrix multiplication} ($\OSMM$), the problem of multiplying two $n \times n$ matrices $A, B$ when their product $AB$ is promised to have at most $O(n^{\delta})$ many non-zero entries for a given value $\delta \in [0, 2]$.
		We then show how to speed up these algorithms in the \emph{fully sparse} setting, where the input matrices $A, B$ are themselves sparse. All~of our algorithms work over arbitrary rings.
		
		Our first, deterministic algorithm for $\OSMM$ works via a two-pass reduction to compressed sensing. It runs in roughly $n^{\omega(\delta/2, 1, 1)}$ time, where $\omega(\cdot, \cdot, \cdot)$ is the rectangular matrix multiplication exponent. This substantially improves on prior deterministic algorithms for output-sparse matrix multiplication.
		
		Our second, randomized algorithm for $\OSMM$ works via a reduction to compressed sensing and a variant of matrix multiplication verification, and runs in roughly $n^{\omega(\delta - 1, 1, 1)}$ time.
		This algorithm and its extension to the fully sparse setting have running times that match those of the (randomized) algorithms for $\OSMM$ and $\FSMM$, respectively, in recent work of Abboud, Bringmann, Fischer, and K\"{u}nnemann (SODA, 2024).
		Our algorithm uses different techniques and is arguably simpler.
		
		Finally, we observe that the running time of our randomized algorithm and the algorithm of Abboud et al. are optimal via a simple reduction from rectangular matrix multiplication.
	\end{abstract}

	\pagenumbering{roman}
	\thispagestyle{empty}
	\newpage
	\tableofcontents
	\newpage
	\pagenumbering{arabic}
	
	\section{Introduction}
	\label{sec:intro}
	
	The Matrix Multiplication Problem ($\MM$) is to compute the product $AB$ of two $n \times n$ matrices $A, B$ given as input. It is one of the most fundamental problems in computer science.
	The running time $O(n^{\omega})$ of the fastest algorithm for $\MM$ is governed by the matrix multiplication exponent $\omega \leq 2.372$ (see~\cite{alman2024more} for the current best bound), and it is a major open question whether $\omega = 2$.
	
	While one line of work has sought to improve $\omega$, another has tried to find faster algorithms for important special cases of $\MM$. 
	In this work, we give new algorithms for the \emph{Output-Sparse Matrix Multiplication Problem} (OSMM), where the product $AB$ of the input matrices $A$ and $B$ is promised to be \emph{sparse}.
	Specifically, we consider $\MM$ when $AB$ is promised to have $O(n^{\delta})$ many non-zero entries for some $\delta \in [0, 2]$, which we take to be a parameter of the problem.
	We denote this problem for matrices over a ring $R$ by $\delta$-$\OSMM_R$.
	We also discuss how to extend our algorithms to the \emph{fully sparse} setting and get an additional speed-up for our algorithms in the case where the input matrices $A$ and $B$ are themselves sparse.
	
	\subsection{Our Work}
	\label{sec:our-work}
	
	Our main results are a deterministic algorithm and a randomized algorithm for $\OSMM$ over arbitrary rings $R$ including, in particular, finite fields, the integers, and the real numbers. The algorithms perform a similar number of ring operations regardless of $R$.
	
	Define the rectangular matrix multiplication exponent $\omega(\alpha, \beta, \gamma)$ to be the infimum over values $\omega' \geq 0$ such that there is an algorithm for multiplying an $n^{\alpha} \times n^{\beta}$ matrix and an $n^{\beta} \times n^{\gamma}$ matrix in $O(n^{\omega'})$ time.
	Define the dual matrix multiplication exponent $\omega^{\perp}$ to be the supremum over values $\omega' \geq 0$ such that $\omega(1, 1, \omega') = 2$.
	The theorems corresponding to our main two algorithms are as follows.

	\begin{restatable}[Deterministic OSMM]{maintheorem}{MainDet}
		\label{thm:det-osmm-intro}
		Let $\delta \in [0, 2]$, let $R$ be a ring, and let $\eps > 0$ be a %
		constant.
		There is a deterministic algorithm for solving $\delta$-$\OSMM_R$ on $n \times n$ matrices that performs $O(n^{\omega(\delta/2, 1, 1) + \eps})$ operations over $R$.
	\end{restatable}

	\begin{restatable}[Randomized OSMM]{maintheorem}{MainRand}
		\label{thm:rand-osmm-intro}
		Let $\delta \in [0, 2]$, let $R$ be a ring, and let $\eps > 0$ be a constant.
		There is a randomized algorithm for solving $\delta$-$\OSMM_R$ on $n \times n$ matrices that performs $O(n^{\beta + \eps})$ operations over $R$, where
		\[
		\beta := 
		\begin{cases}
			2 & \text{if $\delta \leq 1 + \omega^{\perp}$ ,} \\
			\omega(\delta - 1, 1, 1) & \text{otherwise .}
		\end{cases}
		\]
	\end{restatable}
	
	Since $\omega^{\perp} \geq 0.321$~\cite{vassilevska-williams2024new}, the algorithms in \cref{thm:det-osmm-intro,thm:rand-osmm-intro} run in essentially quadratic time (which is the best possible without making an assumption about input sparsity) for all $\delta \leq 0.642$ and $\delta \leq 1.321$, respectively. Furthermore, both algorithms run in $O(n^{\omega - \eps})$ time for some $\eps = \eps(\delta) > 0$ when $\delta < 2$.\footnote{Note that $\delta - 1 \leq \delta / 2$ for $\delta \in [0, 2]$ so the randomized algorithm in \cref{thm:rand-osmm-intro} is always at least as fast as the deterministic algorithm in \cref{thm:det-osmm-intro}.}
	We also extend our algorithms to the \emph{fully sparse} setting where the input matrices $A, B$ are themselves sparse (in addition to having the promise that their product $AB$ is sparse), and get faster algorithms in this case; see \cref{sec:fsmm,thm:fsmm-deterministic,thm:fsmm-randomized}.
	Although we choose to highlight our somewhat simpler $\OSMM$ algorithms, they actually follow as corollaries from our algorithms for the fully sparse case.
	
	We note that our extension of \cref{thm:rand-osmm-intro} to the fully sparse setting (given in \cref{thm:fsmm-randomized}) is equivalent to~\cite[Lemma 3.1]{abboud2024time}, the main technical lemma in~\cite{abboud2024time} for fully sparse matrix multiplication on general rings. 
	However, our corresponding algorithm uses different techniques and is arguably simpler.
	In \cref{prop:fsmm-rmm-equiv} we also show that our randomized algorithms in \cref{thm:rand-osmm-intro,thm:fsmm-randomized} (and therefore also~\cite[Lemma 3.1]{abboud2024time}) are \emph{optimal} via a reduction from rectangular matrix multiplication. 
	
	We state the running time bounds of our algorithms in terms of operations over a ring $R$, but these bounds immediately imply algorithms with similar running times (in the usual sense) over many standard domains. Over the integers with entries in $[-M, M]$ each operation requires at most $\poly(\log M, \log n)$ time, and over finite fields $\F_q$ each operations requires at most $\poly(\log q, \log n)$ time. Assigning unit cost to operations on real numbers is standard in the Real RAM model.

	\subsection{Applications and Comparison with Prior Work}
	\label{sec:comparison-prior-work}
	
	Input- and output-sparse matrix multiplication have a number of applications including answering certain database queries~\cite{AP09,GWWZ15,DHK20}, computing transitive closures in sparse graphs~\cite{GWWZ15,BCH15}, multi-source breadth-first search~\cite{ZWWW11,GRS06, GjCHSWLW23}, Markov clustering~\cite{SHAB20}, and error correction of matrix products \cite{GLLPT17,kunnemann2018nondeterministic,roche18}.\footnote{We note that error correction of matrix products is equivalent to $\OSMM$.}
	
	A long line of work has studied both input- and output-sparse matrix multiplication~\cite{Gu78,yuster2005fast,journals/ipl/IwenS09,L09,AP09,P12,kutzkov13,JS15,GWWZ15,kunnemann2018nondeterministic,roche18,DHK20,abboud2024time}.
	Our deterministic algorithm in \cref{thm:det-osmm-intro} substantially improves on prior deterministic $\delta$-$\OSMM$ algorithms for most values of $\delta$, and is never slower than them. See~\cref{fig:osmm-exp-bounds}.
	As noted above, our algorithm runs in essentially optimal quadratic time for $\delta \leq 0.642$, and for every constant $\delta < 2$ it runs in $O(n^{\omega - \eps})$ time for some $\eps = \eps(\delta) > 0$.
	
	The main prior works on deterministic $\OSMM$ are Kutzkov~\cite{kutzkov13}, which gives a $O(n^2 + n^{2 \delta + 1})$-time algorithm for $\OSMM$ over the reals, and K\"{u}nnemann~\cite{kunnemann2018nondeterministic}, which gives a  $O(n^{2 + \delta/2} + n^{2 \delta})$-time algorithm for $\OSMM$ over the integers.
	Their algorithms only run in $O(n^{\omega - \eps})$ time for $\delta < (\omega - 1)/2 \approx 0.686$ and $\delta < 2 (\omega - 2) \approx 0.744$, respectively.
	(One can also show analytically that \cref{thm:det-osmm-intro} is at least as fast as~\cite{kutzkov13,kunnemann2018nondeterministic} for all $\delta$.\footnote{By \cref{thm:rectMM-ub} we have that $\omega(1, 1, \delta/2) \leq 1.825 + 0.548 \delta/2$ from which it follows that $\omega(1, 1, \delta/2) \leq 2 \delta + 1$ for all $\delta \geq 0.5$. So, our algorithm is always at least as fast as the algorithm in~\cite{kutzkov13}.
		Furthermore, by \cref{prop:alphaProps}, \cref{item:alpha0}, $\omega(\delta/2, 1, 1) \leq 2 + \delta/2$ and so our algorithm is always at least as fast as the algorithm in~\cite{kunnemann2018nondeterministic}.})
	
	Finally, we note a connection between our algorithm for $\OSMM$ and an algorithm from~\cite{BGGW24-MMV} for the All Zeroes Problem, the problem of checking whether $AB = 0$ for two $n \times n$ input matrices $A, B$.
	(We note that~\cite{kunnemann2018nondeterministic} showed that the All Zeroes Problem and the better-known Matrix Multiplication Verification Problem ($\MMV$) are essentially equivalent.)
	Specifically, \cref{thm:det-osmm-intro} achieves the same running time of roughly $n^{\omega(\delta/2, 1, 1)}$ for $\delta$-$\OSMM$ that a deterministic algorithm in \cite{BGGW24-MMV} achieved for the All Zeroes Problem with the same promise that the input matrices $A, B$ satisfy $\norm{AB}_0 \leq O(n^{\delta})$.
	The All Zeroes Problem with the guarantee $\norm{AB}_0 \leq O(n^{\delta})$ reduces to $\delta$-$\OSMM$ (this is a trivial decision-to-search reduction), and so one can view \cref{thm:det-osmm-intro} as an upgrade of the deterministic algorithm in~\cite{BGGW24-MMV}.
	We also emphasize again that \cref{thm:det-osmm-intro,thm:rand-osmm-intro} work over \emph{any} ring $R$ whereas the algorithms from each of~\cite{kutzkov13,kunnemann2018nondeterministic,BGGW24-MMV} only work on specific domains.
	
	The sole deterministic $\OSMM$ algorithm that is faster than ours is an algorithm that (only) works for multiplication over nonnegative integers in~\cite[Lemma 3.11]{abboud2024time}.\footnote{We note that algorithms for matrix multiplication over nonnegative integers imply algorithms for Boolean matrix multiplication as a special case. Recall that the $(i,j)$th entry of the Boolean product of matrices $A, B \in \bit^{n \times n}$ is defined as $\lor_{k = 1}^n (A_{i, k} \land B_{k, j})$. To compute this product it suffices to compute the standard matrix product $AB$ over the integers and replace all positive entries of the output with ones.}
	Their algorithm's running time is stated for the fully sparse setting, but it runs in roughly $O(n^{\omega(\delta - 1, 1, 1)})$ time in the output sparse setting. This matches the running time of randomized algorithms in their work and ours for $\OSMM$ over general rings.

	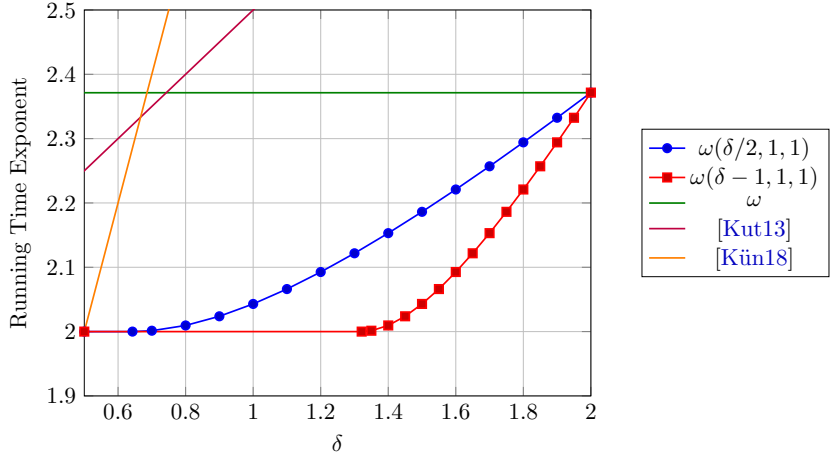
\begin{figure}
		\begin{center}
			\begin{tikzpicture}[scale=0.8]
				\begin{axis}[
					width=10cm,
					height=8cm,
					xlabel={$\delta$},
					ylabel={Running Time Exponent},
					legend style={at={(1.1, 0.5)}, anchor=west, legend columns=1},
					grid=major,
					ymin=1.9, ymax=2.5,
					xmin=0.5, xmax=2,
					enlargelimits=false
					]
					
					\addplot+[thick] coordinates {
						(0.5, 2)
						(0.642668, 2)
						(0.7, 2.001363)
						(0.8, 2.009541)
						(0.9, 2.023788)
						(1, 2.042994)
						(1.1, 2.066134)
						(1.2, 2.092631)
						(1.3, 2.121734)
						(1.4, 2.153048)
						(1.5, 2.18621)
						(1.6, 2.220929)
						(1.7, 2.256984)
						(1.8, 2.294209)
						(1.9, 2.33244)
						(2, 2.371552)
					};
					\addlegendentry{$\omega(\delta/2, 1, 1)$}
					
					\addplot+[thick] coordinates {
						(0.5, 2)
						(1.321334, 2)
						(1.35, 2.001363)
						(1.4, 2.009541)
						(1.45, 2.023788)
						(1.5, 2.042994)
						(1.55, 2.066134)
						(1.6, 2.092631)
						(1.65, 2.121734)
						(1.7, 2.153048)
						(1.75, 2.18621)
						(1.8, 2.220929)
						(1.85, 2.256984)
						(1.9, 2.294209)
						(1.95, 2.33244)
						(2, 2.371552)
					};
					\addlegendentry{$\omega(\delta-1, 1, 1)$}
					
					\addplot[green!50!black, thick]{2.371339};
					\addlegendentry{$\omega$}
					
					\addplot[purple, thick, samples=1000] {
						2 + 0.5*x
					};
					\addlegendentry{\cite{kutzkov13}}
					
					\addplot[orange, thick, samples=1000] { 
						2*x + 1
					};\addlegendentry{\cite{kunnemann2018nondeterministic}}
					
				\end{axis}
			\end{tikzpicture}
		\end{center}
		
		\caption{A plot of the running time exponents of several algorithms for $\delta$-$\OSMM$, with arbitrarily small polynomial factors suppressed.
			The blue and red curves give upper bounds on the running times of our deterministic roughly $n^{\omega(\delta/2, 1, 1)}$-time algorithm in~\cref{thm:det-osmm-intro} and our roughly $n^{\omega(\delta - 1, 1, 1)}$-time randomized algorithm in~\cref{thm:rand-osmm-intro}, respectively. (The latter bound also holds for a randomized algorithm in~\cite{abboud2024time}.)
			The points on the curves come from upper bounds on rectangular $\MM$ exponents in~\cite{vassilevska-williams2024new}, and the line segments connecting them are justified by the convexity of $\omega(\cdot, 1, 1)$.
			We also plot the $\MM$ exponent $\omega$, and the exponents of the $O(n^{\max \set{2, 2 \delta + 1}})$-time algorithm of Kutzkov~\cite{kutzkov13} and the $O(n^{\max \set{2 + \delta/2, 2 \delta}})$-time algorithm of K\"{u}nnemann~\cite{kunnemann2018nondeterministic};~\cite{kutzkov13,kunnemann2018nondeterministic} are the main prior works giving deterministic algorithms for $\OSMM$.}
		\label{fig:osmm-exp-bounds}
	\end{figure}

	\subsection{Technical Overview}
	\label{sec:technical-overview}

	We next give an overview of our algorithms, which use simple primitives: fast rectangular matrix multiplication, a compressed sensing scheme with good parameters, and an algorithm for a variant of matrix multiplication verification.
	Define $\norm{\vec{v}}_0$ (respectively, $\norm{M}_0$) to be the number of non-zero entries in (i.e., Hamming weight of) a vector $\vec{v}$ (respectively, matrix $M$). We call a vector~$\vec{v}$ (respectively, matrix $M$) \emph{$t$-sparse} if $\norm{\vec{v}}_0 \leq t$ (respectively, if $\norm{M}_0 \leq t$).

	\paragraph{Output-column-sparse $\MM$ using compressed sensing.}
	
	In the \emph{compressed sensing} problem with parameters $n, m, t \in \Z^+$ with $1 \leq t \leq n$, the goal is to compute an $m \times n$ matrix $H = H(m, n, t)$ called a \emph{measurement matrix}, and to design an efficient \emph{sparse recovery} algorithm $\cR$ (depending on~$H$) such that for every $t$-sparse vector $\vec{x}$ of length $n$, $\cR(H \vec{x}) = \vec{x}$. That is, the goal is for $\cR$ to be able to recover any $t$-sparse vector $\vec{x}$ from $H \vec{x}$.
	
	Several previous works including \cite{journals/ipl/IwenS09} use compressed sensing to solve 
	$\MM$ in the case where every column $\vec{c}_i$ of the output matrix $C = (\vec{c}_1, \ldots, \vec{c}_n) := AB$ is promised to be $t$-sparse.
	The technique works by computing $HAB$ parenthesized as $(HA) B$ using fast rectangular matrix multiplication, and then recovering the columns $\vec{c}_i$ of $C = AB$ from $HC$ using $\cR$ as $\vec{c}_i = \cR(H \vec{c}_i)$ for $i = 1, \ldots, n$.
	If computing $H$ takes $T_{\cM}(n)$ time, $\cR$ takes $T_{\cR}(n)$ time, and $m = O(n^{\beta})$ for some $\beta \in [0, 1]$, then for any constant $\eps > 0$ the overall running time of this algorithm is
	\begin{equation} \label{eq:osmm-via-cs-runtime}
		O(T_{\cM}(n) + n^{\omega(\beta, 1, 1) + \eps} + n \cdot T_{\cR}(n)) \ \text{.}
	\end{equation}
	
	We note that~\cite{journals/ipl/IwenS09} uses the compressed sensing scheme of~\cite{conf/soda/Indyk08}. The compressed sensing scheme there has good parameters and has a fast sparse recovery algorithm $\cR$. However, while~\cite{conf/soda/Indyk08} uses a polynomial-time construction of expanders, it is unclear that its construction runs in better-than-$n^{\omega}$ time. (We achieve $T_{\cM}(n) = \Ot(n)$ time in \cref{thm:main-compressed-sensing}.)
	
	\paragraph{Reducing $\OSMM$ to compressed sensing.} 
	
	As our first contribution, we show how to replace the assumption in the above compressed-sensing-based algorithm for $\MM$ that every column of $C = AB$ is $t$-sparse with the ``global'' assumption that $C$ is $t^2$-sparse. If $t^2 = O(n^{\delta})$, this is exactly the assumption in $\delta$-$\OSMM$.
	To do this, we use a two-pass reduction, which roughly speaking first computes a matrix $C'$ that agrees with $C$ on sparse columns and then computes all rows of $C - C'$.
	This reduction appears formally as~\cref{alg:osmm-deterministic}.
	
	Assume without loss of generality that $\cR$ always outputs \emph{some} vector of length $n$ even when its input is not $t$-sparse.
	More specifically, assume that $\vec{c}_i' := \cR(H \vec{c}_i)$ is always a vector of length $n$ with the guarantee that $\vec{c}_i' = \vec{c}_i$ if $\vec{c}_i$ is $t$-sparse.
	Let $C' := (\vec{c}_1', \ldots, \vec{c}_n')$. The key observation is that if $C = AB$ is $t^2$-sparse, then \emph{every} row of $C - C'$ must be $t$-sparse. To see this, note that if $C$ is $t^2$-sparse, then it has fewer than $t$ columns that are \emph{not} $t$-sparse, and that these are the only (potentially) non-zero columns in $C - C'$. %
	We can then recover the entries of $C - C'$ (and therefore also $C$, since we know $C'$) by computing $H(C - C')$ as $(HA)B - HC'$ and running $\cR$ on each row of the result. %
	
	\paragraph{Reducing $\OSMM$ to compressed sensing and a variant of $\MMV$.}
	
	Our second reduction is from $\OSMM$ to compressed sensing and a variant of the matrix multiplication verification problem ($\MMV$) that we call $\ColMMV$.
	The goal of $\ColMMV$ is, given three matrices $A$, $B$, and $C'$ as input, to identify all of the indices $i \in [n]$ such that $\vec{c}_i \neq \vec{c}_i'$, where $C = (\vec{c}_1, \ldots, \vec{c}_n) := AB$ and $C' = (\vec{c}_1', \ldots, \vec{c}_n')$.
	This reduction appears formally as \cref{alg:osmm_via_mmv}.
	
	The main idea of our second reduction is to sequentially ``guess and check'' all columns of $C$ of sparsity at most $2^i$ for $i = 0, \ldots, \ceil{\log n}$.
	Namely, in iteration $i$ we use a compressed sensing scheme as above to compute a matrix $C'$ that agrees with all columns of $C$ of sparsity at most $t = 2^i$.
	(As before, we assume that $\cR$ always outputs some vector of length $n$ regardless of its input.)
	We then use our $\ColMMV$ oracle to identify which of the columns of $C$ were computed correctly in $C'$. We use this information to update the corresponding columns of our output matrix and to remove the corresponding columns from $B$. We emphasize that this reduction itself is deterministic, and the only use of randomness in the algorithm summarized in \cref{thm:rand-osmm-intro} comes from solving $\ColMMV$.
	
	Correctness follows simply by analyzing the last iteration (in which $i = \ceil{\log n}$) since each column of $C$ has sparsity at most $n \leq 2^{\ceil{\log n}}$.
	In fact, the reduction does not require an upper bound on the sparsity of $C$ as input and works regardless of this sparsity---it is a reduction from $\MM$, not just $\OSMM$.
	However, the sparser $C$ is the faster the reduction is.
	Let $J_{i + 1} \subseteq [n]$ be the set of indices of columns of $C$ with more than $2^i$ non-zero entries.\footnote{In \cref{alg:osmm_via_mmv}, $J_{i + 1}$ can technically be a subset of these columns if the sparse recovery algorithm happens to succeed (i.e., when $\cR(H\vec{x}) = \vec{x}$) at an earlier iteration $j < i$ even when $\vec{x}$ is not $2^j$-sparse. But, $\card{J_{i + 1}}$ being smaller only makes the algorithm run faster.}
	
	Assume that it is possible to multiply the $m \times n$ measurement matrix $H$ by an arbitrary $n \times n$ matrix $A$ in $O(n^{2 + \eps})$ time for any constant $\eps > 0$. When we instantiate the reduction with the concrete compressed sensing scheme described in \cref{thm:main-compressed-sensing} this will be the case since our measurement matrices $H$ will be sparse.
	Then the bottleneck in each iteration $i$ is computing the product of $HA$ (an $m \times n$ matrix) and $B$ (an $n \times \card{J_i}$ matrix).
	Let $\alpha := i/\log n$. We know that $\card{J_{i + 1}} \leq O(n^{\delta}/2^i) = O(n^{\delta - \alpha})$, and if $m \approx t = 2^i = n^{\alpha}$ then it is possible to compute the product of $HA$ and $B$ in roughly $n^{\omega(\alpha, 1, \delta - \alpha)}$ time.
	By \cref{cor:alphaFla}, $\omega(\alpha, 1, \delta - \alpha) \leq \omega(\delta - 1, 1, 1)$ for any $\alpha$, and so the claimed running time of $O(n^{\omega(\delta - 1, 1, 1) + \eps})$ follows since we only need to compute $O(\log n)$ such matrix products and the multiplicative $\log n$ factor gets absorbed by the $\eps$ in the exponent.
	
	\paragraph{From Output-Sensitive $\MM$ to Fully Sparse $\MM$.}
	In \cref{sec:fsmm}, we extend our algorithms to the \emph{fully sparse} case where we consider the sparsity of the $n \times n$ input matrices $A, B$ in addition to the sparsity of $AB$. I.e., we are promised not only that $\norm{AB}_0 \leq O(n^{\dout})$ for some $\dout \in [0, 2]$, but also that $\norm{A}_0, \norm{B}_0 \leq O(n^{\din})$ for $\din \in [0, 2]$.\footnote{Without loss of generality we can take $\dout \leq 2 \din$.}
	For $\alpha, \beta, \gamma \in [0, 1]$, define $\omega_{\din}(\alpha, \beta, \gamma)$ to be the infimum over $\omega' \geq 0$ such that it is possible to compute the product of an $n^{\alpha} \times n^{\beta}$ matrix $A$ and an $n^{\beta} \times n^{\gamma}$ matrix $B$ in $O(n^{\omega'})$ time when $\norm{A}_0, \norm{B}_0 \leq O(n^{\din})$. %
	We note that there are algorithms showing that $\omega_{\din}(\alpha, \beta, \gamma)$ is smaller than $\omega(\alpha, \beta, \gamma)$ for sufficiently small values of $\din$~\cite{yuster2005fast,KSV06}.
	
	\cref{thm:fsmm-deterministic} (respectively, \cref{thm:fsmm-randomized}) extends our deterministic (respectively, randomized) algorithm in \cref{thm:det-osmm-intro} (respectively, \cref{thm:rand-osmm-intro}) to the fully sparse setting.
	\cref{thm:fsmm-deterministic} achieves a running time of roughly $O(n^{\omega_{\din}(\dout/2, 1, 1)} + n^{\dout/2 + 1})$, which is equal to the $O(n^{\omega(\dout/2, 1, 1)})$ running time of \cref{thm:det-osmm-intro} when $\din = 2$, i.e., when there is no promise on the sparsity of $A, B$ as in $\dout$-$\OSMM$.
	
	\cref{thm:fsmm-randomized} achieves a running time of roughly $O(n^{\din} + n^{\beta})$, for $\beta := \max_{\eta \in [0, \min\set{\dout, 1}]} \omega_{\din}(\eta, 1, \dout - \eta)$. Although this expression looks intricate, it appears in~\cite[Lemma 3.1]{abboud2024time} (to which it is equivalent) and arises naturally in our FSMM proof of optimality; see \cref{prop:fsmm-rmm-equiv}. 
	(For $\dout \geq 1$ and $\din = 2$, we also get that $\omega_{\din}(\eta, 1, \dout - \eta) \leq \omega(\dout - 1, 1, 1)$ by \cref{cor:alphaFla}, and moreover $\omega(\dout - 1, 1, 1) \geq 2 \geq \din$. So, we achieve the running time of \cref{thm:rand-osmm-intro}.)
	
	Extending our algorithms to the fully sparse setting essentially only uses basic facts about input-sparse matrix multiplication and properties of our compressed sensing scheme.
	Specifically, it uses the fact that our measurement matrix $H$ has very sparse columns, and that on input $H\vec{x}$ for a $t$-sparse vector $\vec{x}$ of length $n$, $\cR$ runs in time roughly $t$ rather than time roughly $n$.
	In particular, we use that such a matrix $H$ admits fast multiplication by any matrix $A$, and if $A$ is sparse then $HA$ is still nearly as sparse.
	(To achieve all of this, we use sparse representations of matrices and vectors.)
	
	\paragraph{Instantiating our reductions.}
	
	To instantiate our reductions and turn them into the algorithms summarized by \cref{thm:det-osmm-intro,thm:rand-osmm-intro}, we need to give an explicit compressed sensing scheme and an algorithm for $\ColMMV$.
	
	We use a compressed sensing scheme due to Berinde, Gilbert, Indyk, Karloff and Strauss~\cite{BGIKS08}, which requires certain expander graphs. To make the scheme of~\cite{BGIKS08} deterministic and obtain good parameters, we instantiate it with the expander graphs later studied in Guruswami, Umans, and Vadhan~\cite{GUV09}. 
	In particular, our measurement matrices $H$ correspond to the adjacency matrices of the bipartite expanders constructed in~\cite{GUV09} from Parvaresh-Vardy codes. They are sparse and binary.
	The main theorem summarizing this scheme is \cref{thm:main-compressed-sensing}.
	
	For completeness, we present and analyze the expander construction in~\cite{GUV09} and the sparse recovery algorithm in~\cite{BGIKS08} in \cref{sec:compressedSensing}. We show that they are both very efficient. Additionally, we observe that the sparse recovery algorithm in~\cite{BGIKS08} works not just over the real numbers (as it was originally presented), but over arbitrary rings.  
	
	We emphasize the subtlety of identifying an appropriate compressed sensing scheme for our purposes. We need algorithms that compute the measurement matrix and perform sparse recovery not just in (arbitrary) polynomial time, but in nearly quadratic and nearly linear time, respectively. Moreover, we need these algorithms to be deterministic---many works on compressed sensing sample a random measurement matrix $H$ (see, e.g., \cite{journals/tit/CandesRT06,journals/pieee/GilbertI10}).
	Finally, for \cref{thm:rand-osmm-intro} we need multiplication by $H$ to be very efficient.
	
	The algorithm for $\ColMMV$ that we use (stated formally in \cref{lem:freivalds-for-colMMV}) is a modified version of Freivalds's algorithm for ``plain'' $\MMV$~\cite{freivalds1979fast}. This algorithm is the only use of randomness in our our randomized $\OSMM$ algorithm, and so if it could be derandomized efficiently then we would get a \emph{deterministic} algorithm for $\delta$-$\OSMM$ running in roughly $n^{\omega(\delta - 1, 1, 1)}$ time.
	(Similarly, this algorithm is the only use of randomness in our randomized $\FSMM$ algorithm.)
	
	\paragraph{An alternative randomized algorithm.} We note that it is possible to implement a variant of our randomized algorithm that achieves a similar running time to \cref{thm:rand-osmm-intro} if it is possible to approximate the number of non-zero entries in each column of $C = (\vec{c}_1, \ldots, \vec{c}_n) = AB$ to within a constant factor efficiently.
	
	In fact, we can compute estimates $\ct_i$ of the sparsity of each column $\vec{c}_i$ satisfying
	$
	1/2 \cdot \norm{\vec{c}_i}_0 < \ct_i < 2 \cdot \norm{\vec{c}_i}_0
	$
	with high probability using a randomized algorithm from~\cite{PS14}.
	We can then use these estimates to multiply $A$ and $B$ efficiently. We first partition the columns of $B$ (and therefore $AB$) into blocks $B_j$, where $B_j$ consists of all columns $\vec{b}_i$ such that $2^{j-1} \leq \ct_i < 2^j$ for $j = 1, \ldots, r$, where $r \approx \log_2 n$. We then compute $A B_j$ for each $j$ using the algorithm for output-column-sparse matrix multiplication sketched earlier using the assumption (which holds with high probability) that each column of $A B_j$ is $2^{j+1}$-sparse.
	
	We note that the column-sparsity-estimation-based algorithm sketched here uses randomness in a different way from \cref{thm:rand-osmm-intro}; it uses randomness to estimate column sparsities and does not need to solve $\ColMMV$. However, we believe that the reduction from $\OSMM$ to $\ColMMV$ that we use to prove \cref{thm:rand-osmm-intro}---which can be viewed as a search-to-decision reduction of sorts---is elegant and more amenable to derandomization. In particular, a strong resolution to the famous question of derandomizing Freivalds's algorithm would derandomize \cref{thm:rand-osmm-intro} as well.
	
	\subsection{Acknowledgments}
	HB and AG would like to thank Divesh Aggarwal, the National University of Singapore (NUS), and the Centre for Quantum Technologies (CQT). Some of this work was completed while visiting them. The authors would also like to thank Thatchaphol Saranurak for suggesting the use of the expander construction in~\cite{GUV09}, Noah Stephens-Davidowitz for helpful comments, and the anonymous reviewers of a prior version of this work for pointing out the connection between~\cite[Lemma 3.1]{abboud2024time} and our work.

	\section{Preliminaries}
	\label{sec:prelims}
	By $\Z^+$ we denote the set of positive integers. For~$n\in\Z^+$, $[n]$ denotes the set $\{1,\ldots,n\}$. Let $R$ be a ring, $\vec{v}\in R^n$ be a vector, and $M\in R^{m\times n}$ be a matrix. Then by $\norm{\vec{v}}_0$ and $\norm{M}_0$ we denote the number of non-zero entries in~$\vec{v}$ and $M$, respectively. We say $\vec{v}$ (respectively, matrix~$M$) is \emph{$t$-sparse} if $\norm{\vec{v}}_0 \leq t$ (respectively, if $\norm{M}_0 \leq t$). We say that $M$ is \emph{$t$-column-sparse} if each of its columns is $t$-sparse. For an $n\in\Z^+$, $\vec{0}_n$ and $\vec{1}_n$ denote $n$-dimensional vectors with all zeros and all ones, respectively. 
	
	We use $\log(\cdot)$ to denote the logarithm base~2, i.e., $\log(2^n) = n$. We use the notation $\widetilde{O}(\cdot)$ to suppress polylogarithmic factors in the argument.

	\subsection{Matrix Multiplication}
	We first define the main two variants of matrix multiplication that we study.
	
	\begin{definition}[Output and Fully Sparse Matrix Multiplication] \label{def:fsmm}
		Let $R$ be a ring, and let $\din, \dout \in [0, 2]$.
		We define the $(\din, \dout)$-\emph{Fully Sparse Matrix Multiplication Problem} over $R$ ($(\din, \dout)$-$\FSMM_R$) as follows.
		The input consists of $A, B \in R^{n \times n}$ for some $n \in \Z^+$ with $\norm{A}_0, \norm{B}_0 \leq O(n^{\din})$ satisfying the promise that $\norm{AB}_0 \leq O(n^{\dout})$.
		The goal is to compute $AB$.
		
		We additionally define the $\dout$-\emph{Output-Sparse Matrix Multiplication Problem} over $R$ ($\dout$-$\OSMM_R$) as $(2, \dout)$-$\FSMM_R$.
	\end{definition}
	
	When both $\din = 2$ and $\dout = 2$ this corresponds to the problem of matrix multiplication in its general form. 
	Furthermore, note that since the product of two $t$-sparse matrices is $t^2$-sparse, we may without loss of generality assume that $\dout \leq 2 \din$.
	
	We next define rectangular matrix multiplication exponents, input-sparse rectangular matrix multiplication exponents, and the dual matrix multiplication exponent.
	
	\begin{definition}[Rectangular and Input-Sparse Rectangular Multiplication Exponents] \label{def:MM-exponents}
		For constants $\alpha,\beta,\gamma\geq 0$, the \emph{rectangular matrix multiplication exponent $\omega(\alpha,\beta,\gamma)$} is the infimum over all $\omega'\geq 0$ such that the product of $A\in\Z^{n^\alpha \times n^\beta}$ and $B\in\Z^{n^\beta \times n^\gamma}$ can be computed using $O(n^{\omega'})$ arithmetic operations. 
		
		We extend this notation and use $\omega_{\din}(\alpha, \beta, \gamma)$ to denote the case when both $A$ and $B$ are promised to be $O(n^{\din})$-sparse for $0 \leq \din \leq \max \set{\alpha + \beta, \beta + \gamma}$.
	\end{definition}
	
	We note that the \emph{square matrix multiplication exponent} $\omega$ is simply $\omega:=\omega(1,1,1)$, and that $\omega_{\din}(\alpha, \beta, \gamma) \leq \omega_{\max \set{\alpha + \beta, \beta + \gamma}}(\alpha, \beta, \gamma) = \omega(\alpha, \beta, \gamma)$ for all $\alpha, \beta, \gamma \geq 0$.

	\begin{definition}[Dual Matrix Multiplication Exponent]
		The \emph{dual matrix multiplication exponent~$\dualMMexp$} is defined as 
		\[
		\dualMMexp := \sup \set{\omega' \geq 0 : \omega(1, 1, \omega') = 2} \;.
		\]
	\end{definition}
	A beautiful line of work~\cite{gall2018improved,alman2021refined,duan2022faster,gall2023faster,vassilevska-williams2024new,alman2024more} has established strong bounds on multiplication exponents, in particular showing that
	\begin{align*}
		\omega & < 2.371339 \text{\;\;\;\cite{alman2024more}\;,}\\
		\dualMMexp & \geq 0.321334 \text{\;\;\;\cite{vassilevska-williams2024new}\;.}
	\end{align*}
	
	\subsection{Properties of Matrix Multiplication Exponents}
	Next we list several properties of rectangular matrix multiplication exponents $\omega(\cdot, \cdot, \cdot)$.
	\begin{proposition}[{\cite{lotti1983asymptotic}}]\label{prop:alphaProps}
		Let $\alpha_1, \alpha_2, \alpha_3, \beta_1, \beta_2,\beta_3, \lambda \geq 0$.
		The following hold:
		\begin{enumerate}
			\item\label{item:alpha0} $\max(\alpha_1+\alpha_2,\alpha_1+\alpha_3, \alpha_2+\alpha_3)\leq \omega(\alpha_1,\alpha_2,\alpha_3)\leq\alpha_1+\alpha_2+\alpha_3$.
			\item\label{item:alpha2} $\omega(\lambda \alpha_1, \lambda \alpha_2, \lambda \alpha_3) = \lambda \cdot \omega(\alpha_1, \alpha_2, \alpha_3)$.
			\item\label{item:alpha3} $\omega(\alpha_1+\beta_1, \alpha_2+\beta_2,\alpha_3+\beta_3)\leq \omega(\alpha_1, \alpha_2, \alpha_3) + \omega(\beta_1, \beta_2, \beta_3)$.
			\item\label{item:alpha1} For any permutation $\pi : [3] \to [3]$,
			$\omega(\alpha_1, \alpha_2, \alpha_3) = \omega(\alpha_{\pi(1)}, \alpha_{\pi(2)}, \alpha_{\pi(3)})$.
			\item $\omega(\alpha,\beta,\gamma)$ is continuous and non-decreasing in its arguments for $\alpha,\beta,\gamma\geq0$.
		\end{enumerate}
	\end{proposition}
	
	We will make use of the following upper bound on $\omega(\alpha, 1, 1)$.
	\begin{theorem}[{\cite{lotti1983asymptotic,le2012faster}}] \label{thm:rectMM-ub}
		Let $\alpha \in [0, 1]$. Then
		\[
		\omega(\alpha, 1, 1) \leq
		\begin{cases}
			2 & \text{if $0 \leq \alpha \leq \omega^{\perp}$ ,} \\
			2 + \dfrac{\omega - 2}{1 - \omega^{\perp}} \cdot (\alpha - \omega^{\perp}) \leq 1.825+0.548\alpha & \text{if $\omega^{\perp} < \alpha \leq 1$ .}
		\end{cases}
		\]
	\end{theorem}
	
	We next give an upper bound on rectangular matrix multiplication exponents $\omega(1, \alpha, \beta)$ with three potentially distinct arguments.
	\begin{corollary}\label{cor:alphaFla}
		Let $\alpha,\beta,\gamma\in[0,1]$ be such that $\beta+\gamma\geq 1$. Then $\omega(\alpha,\beta,\gamma) \leq \omega(\alpha,1,\beta+\gamma-1)$.
	\end{corollary}
	\begin{proof}
		If $\beta=\gamma=1$, the statement holds trivially. Otherwise we can assume that $\beta+\gamma<2$. Let $\lambda = (1-\gamma)/(2-\beta-\gamma)\geq0$. Then we have that
		\begin{align*}
			\omega(\alpha,\beta,\gamma)&=\omega(\lambda\alpha + (1-\lambda)\alpha, \lambda + (1-\lambda)(\beta+\gamma-1), \lambda(\beta+\gamma-1)+(1-\lambda))\\
			& \leq \omega(\lambda\alpha, \lambda, \lambda(\beta+\gamma-1))+\omega((1-\lambda)\alpha, (1-\lambda)(\beta+\gamma-1), (1-\lambda))\\
			& = \lambda\cdot\omega(\alpha,1,\beta+\gamma-1)+(1-\lambda)\cdot\omega(\alpha,\beta+\gamma-1,1)\\
			& = \omega(\alpha,1,\beta+\gamma-1) \;,
		\end{align*}
		where the inequality uses \cref{prop:alphaProps},~\cref{item:alpha3}, the following equality uses \cref{prop:alphaProps},~\cref{item:alpha2}, and the last one uses \cref{prop:alphaProps},~\cref{item:alpha1}.
	\end{proof}

	\subsection{Input-Sparse Matrix Multiplication}
	In this paper, we assume \emph{sparse representations} of vectors and matrices. Specifically, we represent a vector $\vec{x} \in R^n$ by a list of index-value pairs of its non-zero entries. We represent a matrix $A=(\vec{a}_1,\ldots,\vec{a}_n)\in R^{m\times n}$ using lists of the non-zero entries of each column $\vec{a}_i$, sorted by~$i$.

	We next give a basic algorithm for matrix multiplication when one of the matrices is column-sparse. This algorithm is folklore, and appears in work at least as early as Gustavson~\cite{Gu78}.
	\begin{proposition}[Na\"{i}ve Input-Sparse Matrix Multiplication]
		\label{prop:naive-sparse-mm}
		Let $A\in R^{m \times n}, B \in R^{n \times p}$ for a ring $R$ and $m,n,p \in \Z^+$. Let $A$ be $t$-column-sparse for $t \in \Z^+$. Then $\norm{AB}_0 \leq t\cdot \norm{B}_0$, and there is a deterministic algorithm that computes $AB$ using  $t \cdot \norm{B}_0 \cdot \poly(\log{n})$ arithmetic operations over~$R$.
	\end{proposition}
	\begin{proof}
		For $i\in[n]$, let $b_i$ be the number of non-zero elements in the $i$th row of~$B$. If we consider the multiplication of $A$ and $B$ as a sum of outer products, then for every $k \in [n]$, the outer product of the $k$th column of $A$ and the $k$th row of $B$ can be computed using at most $t \cdot b_k$ multiplications, and has sparsity at most $t \cdot b_k$. Therefore, $\norm{AB}_0\leq \sum_{i\in[n]} t \cdot b_k=t\cdot\norm{B}_0$. The running time of this algorithm is bounded by $t \cdot \norm{B}_0 \cdot \poly(\log{n})$, where the $\poly(\log{n})$ factor accounts for the sparse representation.
	\end{proof}
	
	Next, we present known bounds on the complexity of matrix multiplication when the input matrices are sparse. The following bound was proven for square matrices in~\cite{yuster2005fast}, and was generalized to rectangular matrices in~\cite{KSV06} (see also~\cite[Lemma~4.2]{abboud2024time}).
	\begin{theorem}[\cite{yuster2005fast,KSV06}]\label{thm:yz}
		Let $\alpha,\beta,\gamma>0$ and $0 \leq \din \leq \max \set{\alpha + \beta, \beta + \gamma}$. Then
		\[
		\omega_\din(\alpha,\beta,\gamma)\leq
		\min_{\eta\in[0,\beta]} \max\set{\omega(\alpha, \eta, \gamma),\, 2\din-\eta} \;.
		\]
	\end{theorem}
	
	Next, we list two additional bounds on sparse rectangular matrix multiplication exponents $\omega_\din(\cdot,\cdot,\cdot)$.
	\begin{proposition}\label{prop:omega_din}
		Let $\alpha,\beta,\gamma>0$ and $0 \leq \din \leq \max \set{\alpha + \beta, \beta + \gamma}$. The following hold:
		\begin{enumerate}
			\item\label{item:omega_din_ub}
			$\omega_\din(\alpha,\beta,\gamma)\leq \din+\min \set{\alpha, \gamma}$.
			\item\label{item:omega_din_lb}
			If $\din \geq \max\set{\alpha, \gamma}$, then
			$\omega_\din(\alpha,\beta,\gamma)\geq \alpha+\gamma$.
		\end{enumerate}
	\end{proposition}
	\begin{proof}
		The first bound follows from \cref{prop:naive-sparse-mm} by noting that every $n^\alpha \times n^\beta$ matrix is $n^\alpha$-column-sparse (and that the transpose of an $n^\beta \times n^\gamma$ matrix is $n^\gamma$-column-sparse).
		
		For the second bound, consider $A=(\vec{1},\vec{0},\ldots,\vec{0})\in R^{n^\alpha \times n^\beta}$ and $B=(\vec{1},\vec{0},\ldots,\vec{0})^T \in R^{n^\beta \times n^\gamma}$. Note that $\norm{A}_0,\norm{B}_0 \leq O(n^\din)$, and $\norm{AB}_0=n^{\alpha+\gamma}$. Therefore, any algorithm computing $AB$ must have running time at least $\Omega(n^{\alpha+\gamma})$.
	\end{proof}

	\section{Output-Sparse Matrix Multiplication}
	\label{sec:osmm}
	
	We now present our main technical results on $\OSMM$. In \cref{sec:compressed-sensing-basics}, we formally define compressed sensing schemes and state our main theorem about a concrete such scheme.
	In \cref{sec:deterministic-osmm}, we give a reduction from $\delta$-$\OSMM$ to compressed sensing and instantiate the reduction to give our roughly $n^{\omega(\delta/2, 1, 1)}$-time deterministic algorithm for $\delta$-$\OSMM$, proving \cref{thm:det-osmm-intro}.
	In \cref{sec:randomized-osmm}, we introduce a variant of matrix multiplication verification ($\MMV$) and give a reduction from $\OSMM$ to compressed sensing and this $\MMV$ variant.
	We then instantiate the reduction to give our roughly $n^{\omega(\delta - 1, 1, 1)}$-time randomized algorithm for $\delta$-$\OSMM$, proving \cref{thm:rand-osmm-intro}.

	\subsection{Compressed Sensing}
	\label{sec:compressed-sensing-basics}
	
	We will use the following formalization of a compressed sensing scheme.
	
	\begin{definition} \label{def:compressed-sensing-scheme}
		Let $R$ be a ring and let $m = m(n, t)$ be a non-decreasing, integer-valued function.
		An $m$-\emph{compressed sensing scheme} over $R$ is a pair of algorithms $(\cM, \cR)$ that work as follows. On~inputs $n$ and $t$ in unary with $1\leq t\leq n$, $\cM$ outputs a matrix 
		$H \in R^{m \times n}$. On input $H \vec{x}$ for a vector $\vec{x} \in R^n$ satisfying $\norm{\vec{x}}_0 \leq t$ 
		(and implicitly~$H$, $n$, and $t$),\footnote{Formally, the algorithm $\cM$ computes a string $\sigma$ as preprocessing that is passed as auxiliary input to $\cR$. Without loss of generality, we assume that $\sigma$ encodes~$H$, $n$, and $t$ (but it could also encode more).} $\cR$ outputs $\vec{x}$. 
	\end{definition}
	
	Here the matrix $H$ output by $\cM$ is called a \emph{measurement matrix}, and $\cR$ is called a \emph{sparse recovery algorithm}.
	We will crucially use the following theorem about an explicit compressed sensing scheme from~\cite{BGIKS08} 
	(which we instantiate with a construction of expanders from~\cite{GUV09}) in both of our main algorithms. 
	The claim in \cref{thm:main-compressed-sensing} follows implicitly by combining~\cite{BGIKS08,GUV09}, but we include a full proof for completeness. However, the proof is somewhat involved and so we defer it to \cref{sec:compressedSensing}.

	\begin{restatable}[\cite{BGIKS08}]{theorem}{MainCS} \label{thm:main-compressed-sensing}
		Let $R$ be a ring and let $\alpha>0$ be a constant. There is a deterministic \mbox{$m$-compressed} sensing scheme $(\cM, \cR)$ over~$R$ with $m = m(n,t)=t^{1+\alpha}\cdot\poly(\log{n})$. The algorithm $\mathcal{M}$ runs in time $\Ot(n)$, and the algorithm $\mathcal{R}$ performs $t^{1+\alpha}\cdot\poly(\log{n})$ arithmetic operations over~$R$.
		
		Furthermore, the measurement matrix $H := \cM(1^n, 1^t)$ is binary and $d$-column-sparse for $d=\poly(\log{n})$, and therefore can be multiplied by an arbitrary vector $\vec{x}\in R^n$ using $\norm{\vec{x}}_0\cdot\poly(\log{n})$ addition operations over~$R$.\footnote{Assuming that~$H$ and~$\vec{x}$ are given in sparse representation by lists of non-zero entries (as in the output of the algorithm~$\cM$).}
	\end{restatable}
	
	\subsection{Deterministic OSMM via a Reduction to Compressed Sensing}
	\label{sec:deterministic-osmm}
	
	In this section, we start by giving our two-pass reduction from $\OSMM$ to compressed sensing in \cref{alg:osmm-deterministic}.
	One may also view this reduction as a meta-algorithm. Making it into an explicit algorithm requires using an explicit compressed sensing scheme (we note this in the ``Subroutine'' heading in the specification of \cref{alg:osmm-deterministic}). 
	
	\begin{algorithm}[th]
		\KwIn{Matrices $A, B \in R^{n \times n}$ over a ring $R$, and $t \in \Z^+$ such that $\norm{AB}_0 \leq t^2$.}
		\KwSubroutine{An $m$-compressed sensing scheme $(\cM, \cR)$ over $R$.}
		\KwOut{The product $C = AB$.}
		\nonl~\\
		Compute a measurement matrix $H \in R^{m \times n} := \cM(1^n, 1^t)$. \\
		Compute $D = (\vec{d}_1, \ldots, \vec{d}_n) := HAB$ using fast rectangular MM. \\
		\DontPrintSemicolon Compute $D' := (\cR(\vec{d}_1), \ldots, \cR(\vec{d}_n))$. \tcp*{\footnotesize{$t$-sparse columns of $C = AB$ agree with $D'$.}}
		Compute $F = (\vec{f}_1, \ldots, \vec{f}_n) := (H (AB - D')^T)$ using fast rectangular MM. \\
		\DontPrintSemicolon Compute $F' := (\cR(\vec{f}_1), \ldots, \cR(\vec{f}_n))$. \tcp*{\footnotesize{$F' = (AB - D')^T$.}}
		\KwRet $D' + (F')^T$.
		\caption{Deterministic output-sparse matrix multiplication over a ring $R$.}
		\label{alg:osmm-deterministic}
	\end{algorithm}

	We next analyze \cref{alg:osmm-deterministic}.
	
	\begin{theorem} \label{thm:osmm-deterministic-correctness}
		Let $R$ be a ring, and let $A, B \in R^{n \times n}$ be matrices satisfying $\norm{AB}_0 \leq t^2$ for $t \in \Z^+$ with $t^2 \leq O(n^{\delta})$ for some $\delta \in [0, 2]$. Let $(\cM, \cR)$ be an $m$-compressed sensing scheme over~$R$ with $m = m(n, t) \leq O(t \cdot n^{\alpha})$ for some $\alpha > 0$.
		Then on input $A, B$%
		, \cref{alg:osmm-deterministic} outputs the matrix product $C := AB$. Furthermore, if $\cM$ runs in time $T_{\cM}(n, t)$ and $\cR$ runs in time $T_{\cR}(n, t)$ then for any constant $\eps > 0$, \cref{alg:osmm-deterministic} performs
		\[
		O(T_{\cM}(n, t) + n \cdot T_{\cR}(n, t) + n^{\omega(\delta/2 + \alpha, 1, 1) + \eps})
		\]
		arithmetic operations over~$R$.
	\end{theorem}
	
	\begin{proof}
		We start by proving correctness. 
		Let $C = (\vec{c}_1, \ldots, \vec{c}_n)$, $D' = (\vec{d}_1', \ldots, \vec{d}_n')$, and $F' = (\vec{f}_1', \ldots, \vec{f}_n')$. 
		We assume without loss of generality that $\cR$ outputs \emph{some} vector in $R^n$ on every input $\vec{y} \in R^m$ (including when $\vec{y} \neq H \vec{x}$ for any $t$-sparse $\vec{x} \in R^n$).
		
		We note that $C$ has at most $t$ columns $\vec{c}_i$ that are not $t$-sparse, since otherwise we would have $\norm{C}_0 > t^2$.
		Furthermore, by the specification of $\cR$, $\vec{d}_i' = \cR(H\vec{c}_i) = \vec{c}_i$ for each $i \in [n]$ such that $\vec{c}_i$ is $t$-sparse. It follows that every row of $C - D'$ is $t$-sparse, and therefore that every column of $(AB - D')^T$ is $t$-sparse. From this and again using the specification of $\cR$, we get that $F' = (AB - D')^T$. Correctness then follows by noting that $C = AB = D' + (F')^T$.
		
		We next analyze the time complexity of \cref{alg:osmm-deterministic}. It makes one call to $\cM$ on input $(1^n, 1^t)$, which takes time at most $T_{\cM}(n, t)$, and $2n$ calls to $\cR$, which takes time at most $2n \cdot T_{\cR}(n, t)$. Furthermore, because $H \in R^{m \times n}$ for $t = O(n^{\delta/2})$ and $m \leq O(t \cdot n^{\alpha}) = O(n^{\delta/2 + \alpha})$, $HAB$ and $H(AB - D')^T$ can be computed in time $O(n^{\omega(\delta/2, 1, 1) + \eps})$ for any constant $\eps > 0$ using fast rectangular matrix multiplication. All other operations run in $O(n^2) \leq O(n^{\omega(\delta/2 + \alpha, 1, 1)})$ time. The result follows by summing the running time bounds for each of these components.
	\end{proof}
	
	Our main theorem about solving $\OSMM$ deterministically follows as a corollary.
	
	\MainDet*
	\begin{proof}
		We instantiate the compressed sensing scheme subroutine needed in \cref{alg:osmm-deterministic} with the scheme described in \cref{thm:main-compressed-sensing}.
		Let $\alpha > 0$ be a constant, which we will set later, and let $\alpha'$ be a constant satisfying $0 < \alpha' < \alpha$.
		From \cref{thm:main-compressed-sensing}, we have that there exists a deterministic $m$-compressed sensing scheme where $m = m(n, t) = t^{1 + \alpha'} \cdot \poly(\log n) \leq O(t n^{\alpha})$, with $T_{\cM}(n, t) \leq T_{\cM}(n, n) \leq \Ot(n) \leq O(n^{1 + \alpha})$ and $T_{\cR}(n, t) \leq T_{\cR}(n, n) \leq O(n^{1 + \alpha})$.
		Furthermore, we have that $\omega(\delta/2 + \alpha, 1, 1) \leq \omega(\delta/2, 1, 1) + \omega(\alpha, 0, 0) =  \omega(\delta/2, 1, 1) + \alpha$ by \cref{prop:alphaProps},~\cref{item:alpha3,item:alpha0}.
		Using the fact that $\omega(\delta/2, 1, 1) \geq 2$ for all $\delta \in [0, 2]$, we then have that for any constant $\eps' > 0$ the algorithm performs at most
		\[
		O(T_{\cM}(n, t) + n \cdot T_{\cR}(n, t) + n^{\omega(\delta/2 + \alpha, 1, 1) + \eps'}) 
		\leq O(n^{1 + \alpha} + n^{2 + \alpha} + n^{\omega(\delta/2, 1, 1) + \alpha + \eps'})
		\leq O(n^{\omega(\delta/2, 1, 1) + \alpha + \eps'})
		\]
		operations over $R$.
		To prove the claim, we need to show that the algorithm performs at most $O(n^{\omega(\delta/2, 1, 1) + \eps})$ operations. This follows by choosing $\eps', \alpha > 0$ such that $\eps' + \alpha \leq \eps$.
	\end{proof}

	\subsection{Randomized OSMM via a Reduction to a Variant of MMV and Compressed Sensing}
	\label{sec:randomized-osmm}
	
	In this section, we get a roughly $n^{\omega(\delta - 1, 1, 1)}$-time algorithm for $\delta$-$\OSMM$ by reducing to a variant of $\MMV$ and to compressed sensing.
	
	\subsubsection{Column-Wise Matrix Multiplication Verification}
	
	In this section, we introduce and give an algorithm for the variant of $\MMV$ in which the goal is to check whether each column of $AB$ is equal to its corresponding column in $C'$.
	
	\begin{definition} \label{def:col-mmv}
		For a ring $R$, the Column-Wise MMV problem over $R$ ($\ColMMV_R$) is defined as follows.
		The input consists of matrices $A \in R^{m \times n}$, $B \in R^{n \times p}$, and $C' = (\vec{c}_1', \ldots, \vec{c}_p') \in R^{m \times p}$.
		Let $C = (\vec{c}_1, \ldots, \vec{c}_p) := AB$.
		The goal is to output the set $J \subseteq [p]$ of indices of columns on which $C$ and $C'$ differ, i.e., 
		\[
		J := \set{j \in [p] : \vec{c}_j \neq \vec{c}_j'} \ \text{.}
		\]
	\end{definition}
	
	We note that ``normal'' $\MMV$ corresponds to deciding whether $J = \emptyset$ in \cref{def:col-mmv}, and so $\MMV$ easily reduces to $\ColMMV$.
	We next argue that a slight variant of Freivalds's algorithm~\cite{freivalds1979fast} also solves $\ColMMV$ in nearly quadratic time.
	
	\begin{lemma}[Freivalds's algorithm for $\ColMMV$]  \label{lem:freivalds-for-colMMV}
		Let $R$ be a ring and let $c \in \Z^+$ be a constant. For any $m, n, p \in \Z^+$, there is a randomized algorithm that solves $\ColMMV_R$ on matrices $A \in R^{m \times n}, B \in R^{n \times p}, C' \in R^{m \times p}$ performing $O(\max \set{mn, np} \cdot \log(mnp))$ arithmetic operations over~$R$ and has a success probability of at least $1 - 1/(mnp)^c$.
		In particular, when $m, p = O(n)$, the algorithm performs $O(n^2 \log n)$ arithmetic operations over~$R$ and has a success probability of at least $1 - 1/n^c$.
	\end{lemma}
	
	\begin{proof}
		The algorithm works as follows. First, it samples a uniformly random matrix $X \sim \bit^{N \times m}$ for $N := \ceil{(c + 1) \log(mnp)}$. It then computes $D = (\vec{d}_1, \ldots, \vec{d}_p) := X(AB - C')$, and outputs $J' := \set{i \in [p] : \vec{d}_i \neq \vec{0}}$. %
		
		We start by showing correctness of the algorithm, i.e., that $J' = J$ for $J$ as defined in \cref{def:col-mmv}.
		For a fixed non-zero vector $\vec{y} \in R^m$ and a uniformly random $\vec{x} \sim \bit^m$, $\Pr[\iprod{\vec{x}, \vec{y}} \neq 0] \geq 1/2$. Indeed, this follows by noting that each term $x_i y_i$ in $\iprod{\vec{x}, \vec{y}} = \sum_{i=1}^m x_i y_i$ for non-zero $y_i$ is equal to $0$ with probability $1/2$ and an induction argument.
		So, for any non-zero column $\vec{y}_i$ of $AB - C'$, we get that $\Pr[X\vec{y}_i = \vec{0}] \leq 1/2^N$. (If $\vec{y}_i = 0$, then $X\vec{y}_i = \vec{0}$ with probability~$1$.)
		Taking a union bound over the at most $p$ non-zero columns of $AB - C'$, we then get that the probability that there exists a non-zero column $\vec{y}_i$ of $AB - C'$ such that $\vec{d}_i = \vec{0}$ is at most $p/2^N \leq p/(mnp)^{c + 1} \leq 1/(mnp)^c$, as needed.
		
		Furthermore, the algorithm runs in $O(N \cdot \max \set{mn, np}) = O(\max \set{mn, np} \cdot \log(mnp))$ time, which finishes the analysis.
	\end{proof}
	
	\subsubsection{The Reduction and its Analysis}
	
	We next give our reduction from $\OSMM$ (and in fact, simply $\MM$) to $\ColMMV$ and compressed sensing in \cref{alg:osmm_via_mmv}. We sometimes use sets to index the rows and columns of a matrix. I.e., for finite sets $I$ and $J$ we use notation like $A \in R^{I \times J}$ to denote a matrix over a ring $R$ with $\card{I}$ rows indexed by $I$ and $\card{J}$ columns indexed by $J$. Similarly, we use notation like $A_{I', J'}$ to denote the submatrix of $A \in R^{I \times J}$ with rows indexed by $I' \subseteq I$ and columns indexed by $J' \subseteq J$.
	
	\begin{algorithm}[th]
		\KwIn{Matrices $A, B \in R^{n \times n}$ over a ring $R$.}
		\KwSubroutine{An $m$-compressed sensing scheme $(\cM, \cR)$ over $R$.}
		\KwSubroutine{An algorithm \textsc{ColMMV} for \ColMMV.}
		\KwOut{The product $AB$.}
		\nonl~\\
		Set $J_0 := [n]$.\\
		Initialize a matrix $C' = (\vec{c}_1', \ldots, \vec{c}_n') \in R^{n \times n}$ arbitrarily.
		
		\For{$i = 0, 1, \ldots, \ceil{\log n}$}{
			Set $t := 2^i$.\\
			Compute $H := \cM(1^n, 1^t) \in R^{m \times n}$. \DontPrintSemicolon \tcp*{\footnotesize{$m = m(n, t)$}.}
			Compute $B' := B_{[n], J_i} \in R^{[n] \times J_i}$. \DontPrintSemicolon \tcp*{\footnotesize{Restrict $B$ to columns indexed by $J_i$.}}
			Compute $D = (\vec{d}_j)_{j \in J_i} := H A B' \in R^{[m] \times J_i}$. \\
			Compute $F = (\vec{f}_j)_{j \in J_i} := (\cR(\vec{d}_j))_{j \in J_i} \in R^{[n] \times J_i}$. \\
			Compute $J_{i + 1} := \textsc{ColMMV}(A, B', F)$. \DontPrintSemicolon 
			\tcp*{\footnotesize{Indices of columns that differ in $AB'$, $F$.}}
			\For{$j \in J_i \setminus J_{i+1}$}{
				Set $\vec{c}_j' := \vec{f}_j$.
			}
		}
		\KwRet $C'$
		\caption{Matrix multiplication using column-wise matrix multiplication verification.}
		\label{alg:osmm_via_mmv}
	\end{algorithm}
	
	In~\cref{thm:osm-via-mmv-correctness}, we analyze \cref{alg:osmm_via_mmv}. We again emphasize that its correctness does not depend on the sparsity of $AB$; only its running time does. Furthermore, the same algorithm works regardless of the sparsity of $AB$.

	\begin{theorem} \label{thm:osm-via-mmv-correctness}
		Let $A, B \in R^{n \times n}$ over a ring $R$, let $(\cM, \cR)$ be an $m$-compressed sensing scheme with $m = m(n, t) \leq O(t n^{\alpha})$ for some $\alpha > 0$, and let \textsc{ColMMV} be an algorithm for $\ColMMV$ over~$R$. Then on input $A, B \in R^{n \times n}$, \cref{alg:osmm_via_mmv} outputs the matrix product $C=AB$.
		
		Furthermore, let $T_{\cM}(n, t)$ and $T_{\cR}(n, t)$ denote the running times of $\cM$ and $\cR$, respectively, let $T_{H}(n, t)$ denote the time needed to multiply $H := \cM(1^n, 1^t)$ by an $n \times n$ matrix, and let $T_{V}(m, n, p)$ denote the running time of \textsc{ColMMV} on input matrices $A \in R^{m \times n}, B \in R^{n \times p}, C \in R^{m \times p}$.
		If~$\norm{AB}_0 \leq O(n^{\delta})$ for $0 \leq \delta \leq 2$ %
		then the algorithm performs
		\[
		\Ot(n^{\beta + \eps} + T_{\cM}(n, n) + n \cdot T_{\cR}(n, n) + T_{H}(n, n) + T_{V}(n, n, n))
		\]
		arithmetic operations over~$R$ for any constant $\eps > 0$,
		where
		\[
		\beta := 
		\begin{cases}
			2 & \text{if $\delta + \alpha \leq 1 + \omega^{\perp}$ ,} \\
			\omega(\delta + \alpha - 1, 1, 1) & \text{otherwise .}
		\end{cases}
		\]
	\end{theorem}
	
	\begin{proof}
		Assume without loss of generality that $\delta \geq 1$. If not, set $\delta = 1$.
		Let $C = (\vec{c}_1, \ldots, \vec{c}_n)$.
		
		We argue that after the $i$th iteration of the for loop, $\vec{c}_j' = \vec{c}_j$ for all $j \in [n]$ such that $\norm{\vec{c}_j}_0 \leq 2^i$.
		Indeed, this follows from the fact that for such an index $j$, $\vec{f}_j = \cR(H \vec{c}_j) = \vec{c}_j$ by the definition of a compressed sensing scheme with $t = 2^i$, and therefore we have that $j \notin J_{i + 1}$.
		(As in the proof of \cref{thm:osmm-deterministic-correctness} we assume that $\cR$ always outputs \emph{some} vector in $R^n$, regardless of its input.)
		In particular, because $\norm{\vec{c}_j}_0 \leq n \leq 2^{\ceil{\log n}}$ for all columns $\vec{c}_j$, $\vec{c}_j' = \vec{c}_j$ for all $j \in [n]$ after the last iteration of the for loop, which implies that $C' = C$, as needed.
		
		We next turn to analyzing the algorithm's running time. We upper bound the worst-case running time of a single iteration (iteration $i$) of the for loop, and note that the loop performs $O(\log n)$ iterations.
		The $\cM(1^n, 1^t)$ call takes time $T_{\cM}(n, t) \leq T_{\cM}(n, n)$, the $\card{J_i} \leq n$ calls to $\cR$ take time at most $\card{J_i} \cdot T_{\cR}(n, t)\leq n\cdot T_{\cR}(n, n)$, and the call to \textsc{ColMMV} takes time at most $T_{V}(n, n, \card{J_i}) \leq T_{V}(n, n, n)$. %
		
		It remains to upper bound the complexity of computing $H A B'$ in the $i$th iteration, which we compute parenthesized as $(H A)B'$.
		Computing $A' := H A$ takes $T_{H}(n, t) \leq T_H(n, n)$ time by assumption,\footnote{When instantiating this reduction each matrix $H$ will be sparse. Such matrices admit fast multiplication.} and it remains to compute $A' B'$, which we do using rectangular matrix multiplication.
		For $i = 0$, we have $t = 1$ and $m = m(n, t) = O(n^{\alpha})$. Therefore, computing $A' B'$ takes $O(n^{\omega(\alpha, 1, 1) + \eps}) = O(n^{\omega(\delta + \alpha - 1, 1, 1) + \eps})$ time.
		For $i \geq 1$, we have $t = 2^i$ and $(t/2) \cdot \card{J_i} \leq \norm{AB}_0$, which holds because $J_i$ consists of indices of columns of $AB$ with more than $t/2$ non-zero entries. 
		
		Let $\eta, \gamma \geq 0$ be such that $t = O(n^{\eta})$ and $\card{J_i} = O(n^{\gamma})$.
		Then because $A' \in R^{m \times n}$ for $m = m(n, t) = O(t n^{\alpha}) = O(n^{\eta + \alpha})$ and $B' \in R^{n \times \card{J_i}}$, $A' B'$ can be computed in $O(n^{\omega(\eta + \alpha, 1, \gamma) + \eps})$ time for any constant $\eps > 0$. Furthermore, we have that $\eta + \gamma \leq \delta$ (since $(t/2) \cdot \card{J_i} \leq \norm{AB}_0 \leq O(n^{\delta})$), and so by \cref{prop:alphaProps}, \cref{item:alpha1} and \cref{cor:alphaFla}, $\omega(\eta + \alpha, 1, \gamma) \leq \omega(\eta + \alpha + \gamma - 1, 1, 1) \leq \omega(\delta + \alpha - 1, 1, 1)$, as needed.
	\end{proof}
	
	Our main theorem about solving $\OSMM$ via a randomized algorithm follows as a corollary.
	
	\MainRand*
	\begin{proof}
		Assume without loss of generality that $\delta \geq 1$. If not, set $\delta = 1$.
		We instantiate the compressed sensing scheme subroutine needed in \cref{alg:osmm-deterministic} with the scheme described in \cref{thm:main-compressed-sensing}.
		The running time analysis follows from a similar argument as in the proof of \cref{thm:det-osmm-intro}, combined with the guarantee that $T_H(n, n) \leq \Ot(n^2)$ from \cref{thm:main-compressed-sensing} and the guarantee that $T_V(n, n, n) \leq \Ot(n^2)$ from \cref{lem:freivalds-for-colMMV}.
		In particular, the fact that $H$ can be multiplied by an arbitrary vector $\vec{x}$ in $\Ot(\norm{\vec{x}}_0)$ time implies that $H$ can be multiplied by an arbitrary matrix $A \in R^{n \times n}$ in $\Ot(n^2)$ time. This implies that $T_H(n, n) \leq \Ot(n^2)$.
		
		Furthermore, the algorithm in \cref{lem:freivalds-for-colMMV} is randomized and succeeds with probability $1 - 1/n^c$ for any constant $c > 0$.
		Fix any such $c > 0$. \cref{alg:osmm_via_mmv} makes $O(\log n)$ calls to $\textsc{ColMMV}$, and so by taking a union bound we get that all of these calls succeed with probability at least $1 - O(\log n/n^c) = 1 - o(1)$. It follows that \cref{alg:osmm_via_mmv} succeeds with probability $1 - o(1)$, as needed.
	\end{proof}

	\section{Fully Sparse Matrix Multiplication}
	\label{sec:fsmm}
	In this section, we describe how to extend our main algorithms for $\OSMM$ to the fully sparse setting when the input matrices $A, B$ are promised not only to satisfy $\norm{AB}_0 \leq O(n^{\dout})$ but also $\norm{A}_0, \norm{B}_0 \leq O(n^{\din})$ for some $\din \in [0, 2]$.
	
	\subsection{Deterministic FSMM}
	\label{sec:deterministic-fsmm}
	
	We first give our deterministic algorithm for $\FSMM$, which follows by instantiating \cref{alg:osmm-deterministic} with the compressed sensing scheme in \cref{thm:main-compressed-sensing} as in the proof of \cref{thm:det-osmm-intro}.
	
	\begin{theorem}[Deterministic FSMM] \label{thm:fsmm-deterministic}
		Let $\din, \dout \in [0, 2]$, let $R$ be a ring, and let $\eps > 0$ be a constant.
		There is a deterministic algorithm for solving 
		$(\din, \dout)$-$\FSMM_R$
		on $n \times n$ matrices that 
		performs $O(n^{\beta + \eps})$ operations over $R$, where
		\[
		\beta := \max \set{\dout/2 + 1, \omega_{\din}(\dout/2,1,1)} \ \text{.}
		\]
	\end{theorem}
	
	\begin{proof}[Proof sketch]
		We refer to \cref{alg:osmm-deterministic}, instantiated with the compressed sensing scheme in \cref{thm:main-compressed-sensing}, and show how to implement its operations under the assumption that $\norm{A}_0, \norm{B}_0 \leq O(n^{\din})$.
		Let $\alpha > 0$ be a small constant to be set in the analysis.
		We first note that $T_{\cM}(n, t) \leq \Ot(n)$ as before.
		
		We now analyze the time needed to compute $D = HAB$. 
		By \cref{thm:main-compressed-sensing}, $H$ is $\poly(\log n)$-column-sparse,  and so it is possible to compute $A' := HA$ in $\Ot(n^{\din})$ time by \cref{prop:naive-sparse-mm}. Since $\norm{A'}_0 = \Ot(n^{\din})$ and $H, A' \in R^{m \times n}$ for $m = O(n^{\dout/2 + \alpha})$, it is possible to compute the product $A' B$ in $O(n^{\omega_{\din}(\dout/2, 1, 1) + \eps})$ time for any constant $\eps > 0$ (assuming that $\alpha > 0$ is chosen to be sufficiently small; see the proof of \cref{thm:det-osmm-intro}). %
		We note that $\omega_{\din}(\dout/2, 1, 1) \geq \din$ (since simply reading in an $n \times n$ matrix $B$ with $\Theta(n^{\din})$ non-zero entries takes $\Omega(n^{\din})$ time), and so computing $HAB$ takes $O(n^{\omega_{\din}(\dout/2, 1, 1) + \eps})$ time overall.
		
		Furthermore, by \cref{thm:main-compressed-sensing},
		$T_{\cR}(t, n) \leq t^{1 + \alpha}\cdot \poly(\log n) \leq \Ot(n^{\dout/2 + \alpha})$ since
		$t = O(n^{\dout/2})$, and so computing $D'$ and $F'$ takes $\Ot(n^{\dout/2 + \alpha + 1})$ time.
		Finally, we analyze the time needed to compute $F$. Computing $H(AB)^T$ can be done in $O(n^{\omega_{\din}(\dout/2, 1, 1) + \eps})$ time for any constant $\eps > 0$ as with computing $D$.
		Furthermore, $D'$ has at most $O(n^{\dout/2})$ non-zero columns, and so, again using the fact that $H$ is $\poly(\log n)$-column-sparse, it takes $\Ot(n^{\dout/2 + 1})$ time to compute $H(D')^T$. The claim follows.
	\end{proof}
	
	We recover \cref{thm:det-osmm-intro} from \cref{thm:fsmm-deterministic} by noting that $\omega_2(\dout/2, 1, 1) = \omega(\dout/2, 1, 1)$ and that $\omega(\dout/2, 1, 1) \geq \dout/2 + 1$ (by \cref{prop:alphaProps},~\cref{item:alpha0}).
	
	Using the bounds on $\omega_{\din}(\dout/2,1,1)$ from \cref{prop:omega_din},~\cref{item:omega_din_ub} and ~\cref{thm:yz}, we conclude that for every $\eps>0$, $(\din, \dout)$-$\FSMM_R$ can be solved using $O(n^{\dout/2+1+\eps}+n^{\dout/2+\din+\eps})$ or $O(n^{\dout/2+1+\eps}+n^{\sigma+\eps})$ operations over~$R$, where
	\begin{align}\label{eq:detFSMMfla}
		\sigma \leq \min_{\eta\in[0,1]} \max\set{\omega(\dout/2, \eta, 1),\, 2\din-\eta} \;.
	\end{align}
	
	\subsection{Randomized FSMM}
	\label{sec:randomized-fsmm}
	
	We now give our randomized algorithm for $\FSMM$ as \cref{thm:fsmm-randomized}. It follows by instantiating \cref{alg:osmm_via_mmv} with the compressed scheme in \cref{thm:main-compressed-sensing} as in the proof of \cref{thm:rand-osmm-intro}.
	We recall that \cref{thm:fsmm-randomized} is equivalent to~\cite[Lemma 3.1]{abboud2024time}, although the algorithm that it is based on uses different techniques.
	
	\begin{theorem}[Randomized FSMM] \label{thm:fsmm-randomized}
		Let $\din\in[1,2], \dout \in [0, 2\din]$, let $R$ be a ring, and let $\eps > 0$ be a constant.
		There is a randomized algorithm for solving
		$(\din, \dout)$-$\FSMM_R$
		on $n \times n$ matrices that performs $O(n^{\din+\eps}+n^{\beta+\eps})$ operations over $R$, where
		\[
		\beta := \max_{\substack{\eta_1, \eta_2 \in [0, 1] \, , \\ \eta_1 + \eta_2 = \dout}} \omega_{\din}(\eta_1,1,\eta_2) \;.
		\]
	\end{theorem}
	
	\begin{proof}[Proof sketch]
		We refer to \cref{alg:osmm_via_mmv}, instantiated with the compressed sensing scheme in \cref{thm:main-compressed-sensing}, and show how to implement its operations under the assumption that $\norm{A}_0, \norm{B}_0 \leq O(n^{\din})$.
		
		As in the proof of \cref{thm:rand-osmm-intro}, we upper bound the worst-case running time of a single iteration (iteration $i$) of the for loop for $i\geq1$ (as the running time for $i=0$ does not affect the final bound).
		Furthermore, as noted there, we have that $t \cdot \card{J_i} = O(n^{\dout})$ at each iteration of the loop. I.e., choosing $\eta_1 > 0$ so that $t = 2^i = O(n^{\eta_1})$, we have that  $\card{J_i} \leq O(n^{\eta_2})$ where $\eta_2\leq\min\set{1, \dout - \eta_1}$.
		We then have that $m = O(t n^{\alpha}) = O(n^{\eta_1 + \alpha})$ for an arbitrarily small constant $\alpha > 0$ to be set in the analysis.
		We note that $T_{\cM}(n, t) \leq \Ot(n)$.
		Moreover, we have that computing $F$ takes $\card{J_i} \cdot T_{\cR}(n, t) \leq n^{n-\eta_1} \cdot t^{1+\alpha} \cdot \poly(\log n) = \Ot(n^{\dout+\alpha})$ time.
		
		We next sketch how to modify the algorithm in the proof of \cref{lem:freivalds-for-colMMV} to run in $\Ot(n^{\din})$ time.
		We note that the matrix $X$ sampled there has $O(\log n)$ rows, and that we need to compute $XAB'$ and $XF$.
		Computing $XA$ and $(XA)B'$ takes $\Ot(n^\din)$ time by \cref{prop:naive-sparse-mm}.
		On the other hand, we have that $\cR$ runs in $t^{1 + \alpha} \cdot \poly(\log n) = \Ot(n^{\eta_1 + \alpha})$ time, and therefore each of the $\card{J_i} = O(n^{\dout - \eta_1})$ columns of $F$ must be $\Ot(n^{\eta_1 + \alpha})$ sparse. Therefore, we can compute $XF$ in time $\Ot(n^{\dout + \alpha})$ time.
		
		It remains to analyze the time complexity of computing $D=HAB'$, since all other operations are efficient.
		We have that $A' := HA$ can be computed in $O(n^{\omega_{\din}(\eta_1 + \alpha, 1, \eta_2) + \eps'})$ time for any constant $\eps' > 0$, which implies that it can be computed in $O(n^{\omega_{\din}(\eta_1, 1, \eta_2) + \eps})$ time for any constant $\eps > 0$ by choosing $\alpha, \eps' > 0$ to be sufficiently small.
		Furthermore, since $H$ is $\poly(\log n)$ column sparse, $A'$ is itself $\Ot(n^{\din})$ sparse. It follows that $HAB = A' B$ can also be computed in $O(n^{\omega_{\din}(\eta_1, 1, \eta_2) + \eps})$ time for any constant $\eps > 0$.

		The claim follows by noting that the for loop has at most $O(\log n)$ iterations, that the definition of $\beta$ uses $\max_{\eta_1,\eta_2 : \eta_1 + \eta_2 = \dout} \omega_{\din}(\eta_1, 1, \eta_2)$, and that $\dout \leq \omega_{\din}(\eta_1, 1, \eta_2)$ for $\eta_1=\min\set{1,\dout}$ and $\eta_2=\dout-\eta_1$ by  \cref{prop:omega_din},~\cref{item:omega_din_lb} since $\eta_1 + \eta_2 = \dout$. Indeed, \cref{prop:omega_din},~\cref{item:omega_din_lb} applies since these values of $\eta_1, \eta_2$ are at most $1$ and $\din \geq 1$.
	\end{proof}
	
	We recover \cref{thm:rand-osmm-intro} from \cref{thm:fsmm-randomized} by setting $\din = 2$ and noting that for $\dout \geq 1$,
	\[
	\max_{\substack{\eta_1, \eta_2 \in [0, 1] \, , \\ \eta_1 + \eta_2 = \dout}} \omega_{\din}(\eta_1,1,\eta_2)
	=
	\max_{\substack{\eta_1, \eta_2 \in [0, 1] \, , \\ \eta_1 + \eta_2 = \dout}} \omega(\eta_1,1,\eta_2)
	\leq \omega(\dout - 1, 1, 1)
	\]
	by \cref{cor:alphaFla}.
	
	While we state \cref{thm:fsmm-randomized} for $\din\geq1$ for simplicity, we remark that the proof can be extended to the case of $\din\in[0,2]$. The main challenge in this case is that we cannot afford to run the algorithm $\cM$ it time $\Ot(n)$ anymore (as this bound might be larger than the claimed upper bound on the running time of the algorithm). To overcome this, we note that the construction of the matrix~$H$ in our compressed sensing scheme in \cref{sec:compressedSensing} is in fact strongly explicit. That is, one can compute all non-zero entries of a given column of~$H$, and therefore multiply by~$H$ efficiently without even running the algorithm~$\cM$ first.

	\subsection{Optimal FSMM (and OSMM)}
	\label{sec:fsmm-optimal}
	
	We now prove that our randomized algorithm for $\FSMM$ in \cref{thm:fsmm-randomized} (which is equivalent to~\cite[Lemma 3.1]{abboud2024time}) is optimal in the sense that any improvement on it would imply a corresponding improvement on input-sparse rectangular matrix multiplication exponents. 
	In fact, we show an \emph{equivalence} between these exponents and the running time exponent for $(\din, \dout)$-$\FSMM_R$.
	We also show an analogous result for the optimality of the algorithm in \cref{thm:rand-osmm-intro} for $\OSMM$.

	\begin{proposition} \label{prop:fsmm-rmm-equiv}
		Let $\din \in [1, 2]$, $\dout\in[0,2\din]$, let $R$ be a ring, and let $\beta \geq 0$. Then there is a randomized algorithm for $(\din,\dout)$-$\FSMM_R$ that runs in time $O(n^{\beta + \eps})$ for every constant $\eps > 0$ if and only if 
		$\beta \geq \max \set{\din, \omega_\din(\eta_1 , 1, \eta_2)}$
		for all $\eta_1, \eta_2 \in [0, 1]$ with $\eta_1 + \eta_2 = \dout$.
		
		Furthermore, if $\dout \in [1, 2]$ then there is a randomized algorithm for $\dout$-$\OSMM$ that runs in time $O(n^{\beta + \eps})$ for every constant $\eps > 0$ if and only if $\beta \geq \omega(\dout - 1, 1, 1)$. 
	\end{proposition}
	
	\begin{proof}
		Let $\eta_1, \eta_2 \in [0, 1]$ with $\eta_1 + \eta_2 = \dout$.
		We reduce the problem of computing the product of $A \in R^{m \times n}$ and $B \in R^{n \times p}$ for $m = O(n^{\eta_1})$ and $p = O(n^{\eta_2})$ when $A$ and $B$ are both $O(n^{\din})$-sparse to $(\din, \dout)$-$\FSMM_R$.
		The reduction simply outputs $A' := (A^T, \vec{0}, \ldots, \vec{0})^T \in R^{n \times n}$ and $B' := (B, \vec{0}, \ldots, \vec{0}) \in R^{n \times n}$. I.e., $A'$ and $B'$ are simply $A$ and $B$ padded with enough rows and columns of $0$s, respectively, to make them square.
		It is easy to recover the product $AB$ from the product $A' B'$. Furthermore, $\norm{A'}_0 = \norm{A}_0 = O(n^{\din})$, $\norm{B'}_0 = \norm{B}_0 = O(n^{\din})$, and, because $AB \in R^{m \times p}$, $\norm{A' B'}_0 \leq mp \leq O(n^{\eta_1 + \eta_2}) = O(n^{\dout})$, as needed.
		Because the reduction worked for arbitrary $\eta_1, \eta_2 \in [0, 1]$ with $\eta_1 + \eta_2 = \dout$, we get that if there is an algorithm using $O(n^{\beta + \eps})$ operations over $R$ for $(\din, \dout)$-$\FSMM_R$ for some $\beta > 0$ and any $\eps > 0$, then $\omega_{\din}(\eta_1, 1, \eta_2) \leq \beta$. Furthermore, we trivially have that such $\beta$ must satisfy $\beta \geq \din$ since it takes $\Omega(n^{\din})$ time to read an $(\din, \dout)$-$\FSMM_R$ instance.
		
		On the other hand, \cref{thm:fsmm-randomized} implies that if $\beta \geq \din$ and $\beta \geq \omega_\din(\eta , 1, \dout-\eta)$ for all $\eta\in[0,\dout]$ then $(\din,\dout)$-$\FSMM_R$ is solvable using $O(n^{\beta + \eps})$ many operations over $R$.
		
		The result for $\OSMM$ follows by combining the padding argument above with $\din = 2$ and $\eta = 1$ together with \cref{thm:rand-osmm-intro}.
	\end{proof}
	
	\newpage
	
	\bibliographystyle{alpha}
	\bibliography{os_mm}
	
	\newpage
	\appendix
	
	\section{Efficient, Deterministic Compressed Sensing over Rings}\label{sec:compressedSensing}
	
	In this section, we present an efficient compressed sensing scheme over an arbitrary ring. Specifically, we show that the compressed sensing scheme from~\cite{BGIKS08}, when instantiated with the expander construction from~\cite{GUV09}, runs efficiently over any ring.
	In~\cref{sec:expanders}, we analyze the time complexity of the expander construction from~\cite{GUV09}. In~\cref{sec:compressed}, we show how to instantiate the compressed scheme from~\cite{BGIKS08} with the constructed expander.
	
	\subsection{Preliminaries}
	We use $\F_q$ to denote a finite field of order $q$.
	Such fields exist for all prime powers $q$.
	Arithmetic operations over $\F_q$, including computing multiplicative inverses, can be performed in time $\poly(\log{q})$ (see, e.g., the excellent textbooks~\cite{shoup2009computational,modernalgebra}).
	
	We will use the following result about deterministically computing field extensions. 
	
	\begin{theorem}[{\cite[Theorem 4.1]{shoup1990new}}] \label{thm:fieldExtension}
		Let $p$ be a prime, let $a, n \in \Z^+$, and let $q := p^a$.
		There is a deterministic algorithm that, given $\F_q$ as input, constructs an irreducible polynomial over $\F_q$ of degree~$n$ in time
		\[
		O(p^{1/2} (\log p)^3 n^{3 + \eps} + (\log p)^2 n^{4 + \eps} + (\log p) n^{4 + \eps} a^{2 + \eps})
		\]
		for any constant $\eps > 0$.
	\end{theorem}

	\subsection{Expanders and Sparse Recovery}\label{sec:expanders}
	We work with undirected bipartite graphs $G=(L \sqcup R, E)$, where $L$ and $R$ denote the two disjoint sets of vertices, and every edge connects a vertex in~$L$ to a vertex in~$R$. We say that $G$ is \emph{$d$-left-regular} if every vertex $v\in L$ has exactly $d$ incident edges. For a subset of vertices $S\subseteq L$, we use $\cN(S)$ to denote the neighborhood of~$S$, i.e., the set of vertices in~$R$ adjacent to at least one vertex in~$S$.
	By the adjacency matrix of~$G$ we mean the matrix $A\in\{0,1\}^{|R| \times |L|}$, where $A[i,j]=1$ if and only if $G$ contains an edge from the $i$th vertex in~$R$ to the $j$th vertex in~$L$. We always represent such a graph~$G$ using sparse representation of its adjacency matrix, i.e., adjacency lists of all left vertices of~$G$. 
	\begin{definition}
		Let $\eps\in[0,1)$ and $K\in\Z^+$. A \emph{$(K, \eps)$-unbalanced expander} is a $d$-left-regular bipartite graph $G = (L\sqcup R, E)$ such that for every $S \subseteq L$ with $\card{S} \leq K$, $\card{\cN(S)} \geq (1- \eps)d\card{S}$.
	\end{definition}
	
	When constructing unbalanced expanders, the goal is to minimize $M$, $d$, and $\eps$, while simultaneously maximizing~$K$. Using the probabilistic method, it is easy to show that for every~$1\leq K \leq  N/2$ and $\eps>0$, a random $d$-left-regular graph is a $(K, \eps)$-expander with $d=O(\log{(N/K)}/\eps)$ and $M=O(Kd/\eps)$.

	However, for our purposes, we require efficient \emph{deterministic} algorithms for constructing expander graphs. This topic has been extensively studied~\cite{TUZ01,CRVW02,GUV09}, with a focus on optimizing expansion properties while maintaining polynomial-time complexity. For our applications, however, it is necessary to construct expanders in deterministic time $\Ot(n^2)$.
	
	To this end, we note that the explicit construction of an unbalanced expander in \cite{GUV09} is in fact extremely efficient. We state their algorithm formally in \cref{alg:expander_construction} and analyze its complexity below in \cref{thm:expander}. Correctness of the construction follows from the list-decoding properties of Parvaresh-Vardy codes \cite{PV05}, and we refer the reader to the original paper for a proof of correctness~\cite[Theorem~3.5]{GUV09}.
	
	\begin{theorem}[\cite{GUV09}]\label{thm:expander}
		Let $\alpha > 0$ and $\eps \in (0, 1)$ be constants. Then for all sufficiently large $N$ and all $K\leq N$, \cref{alg:expander_construction} outputs a $(K, \eps)$-unbalanced expander $G=(L\sqcup R, E)$ with $\card{L}= N$, $\card{R} = O(d^2\cdot K^{1 + \alpha})$, and left degree
		\[d = O\left( \left( \dfrac{\log N \cdot \log K}{\eps} \right)^{1 + 1/\alpha} \right)\;.\]
		The running time of the algorithm is $\Ot(N)$.\footnote{While the stated upper bound on $|R|$ might be larger than $N\cdot\poly(\log{N})$ when $K$ is close to~$N$, $G$ always has $Nd=\Ot(N)$ edges, and therefore the algorithm outputs a sparse representation of the adjacency matrix of $G$ in time~$\Ot(N)$.}
	\end{theorem}

	\begin{algorithm}[th]
		\KwIn{$K,N\in\Z^+$, $K\leq N$, $\alpha>0$, $\eps \in (0,1)$.}
		\KwOut{A $(K, \eps)$-unbalanced expander $G = (L\sqcup R, E)$.}%
		\nonl~\\
		
		Define $n := \ceil{\log N}$, $k := \log K$.\\
		Define $h := \ceil{(2nk/\eps)^{1/\alpha}}$ and $m := \ceil{k / \log h}$.\\
		Define $q := 2^{\floor{\log(h^{1 + \alpha})}}$.\\
		Compute a degree-$n$ irreducible polynomial $p$ over $\F_q$ using \cref{thm:fieldExtension}.\\
		Initialize $L=[N], \; R=\F_q^{m+1}, \; E = \emptyset$.\\
		\For{$\ell=1, \ldots, N$}{
			Set $t$ to be the $\ell$th polynomial over $\F_q$ of degree $\leq(n-1)$. \DontPrintSemicolon \tcp*{\footnotesize{Under arbitrary ordering.}}
			\For{$i=0, \ldots, m-1$}{
				Compute $t_i := t^{h^i} \bmod p$. \DontPrintSemicolon  \tcp*{\footnotesize{Compute degree-$(n-1)$ polynomial over $\F_q$.}}
				
			}
			\For{every element $y \in \F_q$}{
				\For{$j=0, \ldots, m-1$}{
					Compute $e_j := t_i(y)$. \DontPrintSemicolon \tcp*{\footnotesize{Evaluate degree-$(n-1)$ polynomial at $y \in \F_q$.}}
				}
				Define $e:= (\ell, (y, e_0, \ldots, e_{m-1}))\in [N]\times\F_q^{m+1}$. \\
				$E = E \union \{e\}$.
			}
		}
		\KwRet Sparse representation of~$G=(L\sqcup R, E)$.
		\caption{Deterministic Construction of Expanders via Parvaresh-Vardy Codes.}
		\label{alg:expander_construction}
	\end{algorithm}

	\begin{proof}
		The algorithm creates $N$ left vertices, where each vertex corresponds to a univariate polynomial over $\F_q$ of degree at most $n-1$. This is possible because for all large enough $N$, $q\geq 2$, and, thus, the number of such polynomials is %
		$q^n\geq 2^n \geq N$.
		The right vertices correspond to vectors in $\F_q^{m+1}$. Each left vertex is connected to (exactly) $q$ right vertices. Therefore, the number of left vertices is $\card{L}=N$, the degree of each left vertex is 
		\[
		d:= q = \Theta(h^{1+\alpha}) = \Theta\left( \left( \dfrac{\log N \cdot \log K}{\eps} \right)^{1 + 1/\alpha} \right) \;,
		\]
		and the number of right vertices is 
		\[
		\card{R}=q^{m+1} 
		= q^2\cdot q^{m-1}
		\leq q^2\cdot (h^{1+\alpha})^{k/\log{h}} 
		=q^2\cdot K^{1+\alpha}
		\;.
		\]
		
		We now turn to analyzing the running time of the algorithm.
		By \cref{thm:fieldExtension}, we can compute an irreducible polynomial of degree~$n$ over $\F_q$ in deterministic time $q^{1/2} \cdot\poly(n, \log q)$. Each arithmetic operation with univariate degree-$n$ polynomials over $\F_q$ can be performed in time $\poly(n, \log{q})$. 
		
		Next, for each of the~$N$ iterations of the outermost for loop, the algorithm performs the following operations. It computes each of the $m$ polynomials $t_i$ using the standard repeated squaring technique in time $\log(h^i) \cdot \poly(n, \log{q}) \leq m \log h \cdot \poly(n, \log{q})$. Then the algorithm evaluates each~$t_i$ at all points of $\F_q$ in time $q\cdot\poly(n, \log{q})$. Therefore, \cref{alg:expander_construction} lists all edges of~$G$ in time $N \cdot \poly(m, q, k, h, n)=N \cdot \poly(\log{K}, \log{N}) = \Ot(N)$.
	\end{proof}
	
	Although the explicit construction in~\cref{thm:expander} has a slight loss in parameters compared to the previously mentioned randomized construction, its expansion properties still suffice for all of our applications. We will use the following corollary as a building block of the compressed sensing scheme in~\cref{sec:compressed}.

	\begin{corollary}\label{cor:expander}
		Let $\alpha > 0$ be a constant. There is a deterministic algorithm that for all sufficiently large $n \in \Z^+$ and any $t \in \Z^+, t \leq n$, outputs the adjacency  matrix $A\in\{0,1\}^{m\times n}$ of a $(t, 1/12)$-unbalanced expander with $m=t^{1+\alpha}\cdot\poly(\log{n})$. The running time of the algorithm is $\Ot(n)$, and $A$ is $d$-column-sparse for $d=\poly(\log{n})$.
	\end{corollary}
	
	\begin{proof}
		We apply \cref{thm:expander}  with $\eps = 1/12$, $N=n$, and $K=t$. From \cref{thm:expander}, we have that 
		\begin{align*}
			d &=  O\left( \left( \dfrac{\log N \cdot \log K}{\eps} \right)^{1 + 1/\alpha} \right)=\poly(\log(n)) \;,\\
			m &= O(d^2 K^{1+\alpha}) = t^{1+\alpha}\cdot\poly(\log{n}) \;.
		\end{align*}
		Since the sparsity of each column of~$A$ equals the left degree of the expander, we have that each column of $A$ is $d$-sparse.
		Finally, from~\cref{thm:expander}, the running time of the algorithm is~$\Ot(n)$.
	\end{proof}
	
	\subsection{An Efficient Compressed Sensing Scheme over any Ring}\label{sec:compressed}
	We use row tensor Hadamard products of matrices. First, we recall the definition of the Hadamard product, which is simply the coordinate-wise product of two vectors.
	Specifically, the \emph{Hadamard product} of two vectors $\vec{a} = (a_1, \ldots, a_n)$ and $\vec{b} = (b_1, \ldots, b_n)$ in $R^n$ is 
	\[\vec{a} \odot \vec{b}:=(a_1 \cdot b_1, \ldots , a_n \cdot b_n)\in R^n\;.\]
	The row tensor Hadamard product of matrices $A$ and $B$ is the matrix whose rows are the Hadamard products of all pairs of rows from $A$ and~$B$.

	\begin{definition}[Row tensor Hadamard product] The \emph{row tensor Hadamard product} of two matrices $A=(\vec{a}_1,\ldots,\vec{a}_{m_1})^T \in R^{m_1 \times n}$ and $B=(\vec{b}_1,\ldots,\vec{b}_{m_2})^T \in R^{m_2 \times n}$ is the matrix $C = (\vec{c}_1,\ldots,\vec{c}_{m_1m_2})^T:= A \otimes_r B\in R^{(m_1 m_2)\times n}$ defined as follows. For every $i\in[m_1]$ and $j\in[m_2]$,
		\[
		\vec{c}_{(i-1)m_2+j}=\vec{a}_i \odot \vec{b}_j \;.
		\]
	\end{definition}
	
	\begin{algorithm}[!ht]
		\KwIn{$1^n$ and $1^t$, where $n=2^\ell-1$ for large enough $\ell\in\Z^+$, and $t\leq n$.}
		\KwOut{A measurement matrix~$H\in\{0,1\}^{m\times n}$ of an $m$-compressed sensing scheme with $m(n,t)=t^{1+\alpha}\cdot\poly(\log{n})$.}
		\nonl~\\
		Compute $B=(\vec{b}_1,\ldots,\vec{b}_{n})\in\{0,1\}^{\log(n+1)\times n}$, where $\vec{b}_i$ is the binary expansion of $i$. \\
		
		Compute the adjacency matrix $A\in\{0,1\}^{m'\times n}$ of a $(t,1/12)$-unbalanced expander for $m'=t^{1+\alpha}\cdot\poly(\log{n})$ 
		using \cref{cor:expander}.\\
		Compute the measurement matrix $H\in\{0,1\}^{(m'\log(n+1))\times n}$ as $H := A \otimes_r B$. \\
		\KwRet Sparse representation of~$H$.
		\caption{Computing the measurement matrix of the compressed sensing scheme.}
		\label{alg:algM}
	\end{algorithm}

	In this section, we present an efficient $m$-compressed sensing scheme $(\cM, \cR)$ for $m(n,t)=t^{1+\alpha}\cdot\poly(\log{n})$ for an arbitrarily small constant $\alpha>0$. We assume that $n$ is such that $n+1$ is a power of two. The algorithm $\cM$ outputting the measurement matrix~$H\in\{0,1\}^{m\times n}$ is presented in \cref{alg:algM}. The algorithm proceeds as follows. Let $A\in\{0,1\}^{m'\times n}$ be the adjacency matrix of a $(t, 12)$-unbalanced expander. Let $B\in\{0,1\}^{\log(n+1)\times n}$ be the matrix where the $i$th column of $B$ corresponds to the binary expansion of $i$ for $i \in [n]$. Then the measurement matrix is just $H:= A \otimes_r B$.
	
	First, in \cref{lem:reduce} we present a subroutine (\cref{alg:indyk_reduce}) that takes a measurement $H\vec{x}$ for a $t$-sparse vector $\vec{x}\in R^n$ and outputs an ``approximation'' of $\vec{x}$---a vector~$\vec{y}\in R^n$ such that $\norm{\vec{x}-\vec{y}}_0\leq \norm{\vec{x}}_0/2$.

	\begin{algorithm}[!ht]
		\KwIn{Measurement matrix $H:=\cM(1^n,1^t) \in \bit^{m \times n}$ where $t\leq n$ and $n=2^\ell-1$ for $\ell\in\Z^+$, degree parameter~$d$, and measurement $\vec{z}=(z_1, \ldots, z_m):=H\vec{x}$.}
		\KwOut{$\vec{y} \in R^n$ with $\norm{\vec{x}-\vec{y}}_0 \leq \norm{\vec{x}}_0/2$, if $\norm{\vec{x}}_0 \leq t$.}
		\nonl~\\
		$m' = m / \log(n+1)$. \DontPrintSemicolon \tcp*{\footnotesize{$A\in R^{m'\times n}$ is adjacency matrix of expander, and $H := A \otimes_r B$.}}
		$\textsc{Candidates} = [\quad]$\DontPrintSemicolon \tcp*{\footnotesize{List of potential non-zero (coordinate, value) pairs of $\vec{y}$.}}
		Let $\vec{z}:=(\vec{s}^1,\ldots,\vec{s}^{m'})$, where $\vec{s}^i\in R^{\log(n+1)}$ for $i\in[m']$.\\
		\For{$i = 1, \ldots, m'$}{
			\If{$\{\vec{s}^i_1,\ldots,\vec{s}^i_{\log(n+1)}\}\setminus\{0\} =\{v\}$ for some $v\in R$ }{
				\tcp*{\footnotesize{If all non-zero values are the same (and there is one).}}
				$j = \sum_{k \colon s_k=v} 2^{\log(n+1)-k}$. \DontPrintSemicolon \tcp*{\footnotesize{Index $j$ is computed from indices $k$ s.t. $s_k=v$.}}
				Add $(j, v)$ to $\textsc{Candidates}$.
			}
		}
		Initialize $\vec{y}= \vec{0}$.\\

		\ForEach{$(j, v)$ such that \textsc{Candidates} contains $>d/2$ copies of $(j,v)$}{
			
			Set $y_{j} = v$.\\
		}
		\KwRet $\vec{y}$.
		\caption{Reduce algorithm.}
		\label{alg:indyk_reduce}
	\end{algorithm}

	\begin{lemma}[{\cite{BGIKS08}}]\label{lem:reduce}
		Let $n=2^\ell-1$ for $\ell\in\Z^+$, $t\leq n$, $H:=\cM(1^n, 1^t)$ be the matrix output by \cref{alg:algM}, 
		and $\vec{x} \in R^n$ be a $t$-sparse vector. Given $H\vec{x} \in R^m$, \cref{alg:indyk_reduce} outputs a vector $\vec{y} \in R^n$, such that $\norm{\vec{x}-\vec{y}}_0 \leq \norm{\vec{x}}_0/2$. The algorithm performs $m\cdot \poly(\log{n})$ arithmetic operations over~$R$.
	\end{lemma}
	\begin{proof}
		Let $A\in R^{m' \times n}$ be the adjacency matrix of a $(t, \eps=1/12)$-unbalanced expander with left-degree~$d$. Consider a measurement $\vec{z}=(z_1,\ldots,z_m) := H\vec{x}$ of a $t$-sparse vector $\vec{x} = (x_1, \ldots , x_n)$.  Let $X := \textrm{supp}(x)\subseteq[n]$ be the set of (indices of) non-zero coordinates of $\vec{x}$. Similarly, let $\textrm{supp}(A[i])$ be the set of non-zero coordinates in the $i$th row of~$A$. We call $i \in [m']$ an \emph{isolating} measurement for~$j\in X$ if $X \cap \text{supp}(A[i]) = \{j\}$. If~$i$ is an isolating measurement for~$j$, then the $i$th coordinate of the product $(A\vec{x})_i=x_j$ recovers the non-zero entry~$x_j$ of~$\vec{x}$. The key observation of this proof is that when $A$ is the adjacency matrix of an expander, then the set of measurements $A\vec{x}$ isolates most non-zero entries of $\vec{x}$.

		Let us first consider the case where $i$ is an isolating measurement for~$j$. Note that this implies that $x_j \neq 0$. Recall that the measurement matrix $H = A \otimes_r B$, and we will focus on the block of rows of~$H$ with indices from $(i-1)\log(n+1)+1$ to $i\log(n+1)$, which corresponds to our augmentation of the matrix $B$ on row~$i$ of~$A$. Specifically, for $k\in[\log(n+1)]$, let $s_k$ be the inner product of the $(i-1)\log(n+1)+k$th row of~$H$ with~$\vec{x}$, \[s_k = z_{(i-1)\log(n+1)+k}:=(H\vec{x})_{(i-1)\log(n+1)+k}\;.\]

		Since $i$ isolates~$j$, we have that $s_k=x_j$ whenever $B[k,j]=1$, and $s_k=0$ otherwise. Moreover, since the $j$th column of $B$ is the binary expansion of $j$, we can recover the coordinate~$j$ from the indices of non-zero entries from $s_1, \ldots s_{\log n}$. Specifically, when $i$ isolates~$j$, we have that $s_k = x_j$ only if $k$ is set to $1$ in the binary expansion of $j$. 
		Thus, for an $i$ isolating~$j$, we can recover both the non-zero coordinate~$j$ and its value $x_j$.
		
		We now show that since $A$ is the adjacency matrix of a $(t, 1/12)$-unbalanced expander, most non-zero coordinates of~$\vec{x}$ have isolating measurements. We have that each non-zero coordinate of $\vec{x}$ participates in~$d$ measurements, and the maximum number of measurements over non-zero coordinates is $d\norm{\vec{x}}_0$. By the properties of the expander graph, we are guaranteed that the set of non-zero coordinates $X$ participates in at least $(1-\eps)d|X| = (1-\eps)d\norm{\vec{x}}_0$ measurements. This implies that the number of non-isolating measurements is at most $\eps d \norm{\vec{x}}_0 $, and the number of isolating measurements is at least $(1- 2\eps)d\norm{\vec{x}}_0$.
		
		We say that the $i$th coordinate of $\vec{y}$ computed by \cref{alg:indyk_reduce} is \emph{corrupted} if $x_i \neq y_i$. We now bound the number of corrupted indices in~$\vec{y}$. An index~$i$ may be corrupted in one of two ways: (i) $y_i \neq 0$ and $x_i \neq y_i$; (ii) $y_i = 0$ and $x_i \neq 0$. Let $c_1$ denote the number of corrupted indices with $y_i \neq 0$, and $c_2$ denote the number of corrupted indices with $y_i=0$. 
		
		We observe that for the algorithm to set $y_i$ to a non-zero value, it requires the contribution of more than~$d/2$ measurements that can only be used for that index. 
		Thus, we have that 
		$c_1< (\eps d \norm{\vec{x}}_0)/(d/2) \leq 2 \eps \norm{\vec{x}}_0$.
		Next, each of the $c_2$ corrupted indices of the second type must participate in at most~$d/2$ isolating measurements. Then the total number of isolating measurements is at most $d c_2/2 + (\norm{\vec{x}}_0 - c_2) d$. Recall that the expander property implies that the number of isolating measurements is at least $(1-2\eps)d\norm{\vec{x}}_0$.
		Therefore, $(1-2\eps)d\norm{\vec{x}}_0 \leq d c_2/2 + (\norm{\vec{x}}_0 - c_2) d$, which in turn implies that $c_2 \leq 4 \eps \norm{\vec{x}}_0$.
		Thus, the total number of corrupted indices is at most $c_1+c_2\leq 2\eps \norm{\vec{x}}_0 + 4\eps \norm{\vec{x}}_0 \leq \norm{\vec{x}}_0/2$ for $\eps = 1/12$.
		
		We note that for each $i\in[m']$ and each $k\in[\log(n+1)]$, the algorithm performs $\poly(\log{n})$ arithmetic operations over~$R$. This implies that the algorithm outputs a sparse representation of~$y$ using $m\cdot\poly(\log{n})$ arithmetic operations over~$R$.
	\end{proof}

	We are now ready to present the main result of this section---an efficient $m$-compressed sensing scheme for $m(n,t)=t^{1+\alpha}\cdot\poly(\log{n})$ for an arbitrary constant $\alpha>0$.
	
	\begin{algorithm}[!ht]
		\KwIn{Measurement matrix $H:=\cM(1^n,1^t) \in \bit^{m \times n}$ where $n=2^\ell-1$ for $\ell\in\Z^+$, sparsity parameter $t\leq n$, degree parameter~$d$, and measurement $\vec{z}_0:=H\vec{x}  \in R^m$.}
		\KwOut{$\vec{x}$, if $\norm{\vec{x}}_0 \leq t$.}
		\nonl~\\
		\For{$i=1, \ldots, \ceil{\log (2t)}$}{
			Set $\vec{y}_i := \textsc{Reduce}(H, d, \vec{z}_{i-1})$.\\
			\If{$\norm{\vec{y}_i}_0>3t/2$}{
				\KwRet $\vec{0}$.\DontPrintSemicolon \tcp*{\footnotesize{The Recovery procedure failed.}} %
			}
			Set $\vec{z}_i = \vec{z}_{i-1} - H\vec{y}_i$.\\ %
		}
		\KwRet $\sum_{i=1}^{\ceil{\log(2t)}}\vec{y}_i$.
		\caption{Recovery algorithm of the compressed sensing scheme.}
		\label{alg:indyk_recover}
	\end{algorithm}

	\MainCS*
	
	\begin{proof}

		We will assume without loss of generality that $n=2^\ell-1$ for a large enough~$\ell\in\Z^+$.
		We first analyze the algorithm $\cM$ generating the measurement matrix $H\in\{0,1\}^{m\times n}$, which is illustrated in \cref{alg:algM}. The algorithm $\cM$ outputs the measurement matrix $H\in\{0,1\}^{m\times n}$, where $m=m'(\log(n+1))=t^{1+\alpha}\cdot\poly(\log{n})$, as desired.
		
		The algorithm constructs $B$ in time $\Ot(n\log{n})$. Then it constructs the adjacency matrix~$A$ of an expander graph in time $\Ot(n)$ by~\cref{cor:expander}. Recall that by~\cref{cor:expander}, $A$ is $d$-column-sparse for $d=\poly(\log{n})$. Finally, $\cM$ naively computes (a sparse representation of) the row tensor Hadamard product of $A$ and~$B\in R^{\log(n+1)\times n}$ in time $\Ot(\norm{A}_0 \cdot \log(n+1))=\Ot(n)$. Thus, the total running time of $\cM$ is $\Ot(n)$.

		We now analyze the recovery algorithm $\mathcal{R}$ which is illustrated in \cref{alg:indyk_recover}. For any input vector $\vec{z}_0:=H\vec{x}$, the algorithm performs at most $\ceil{\log(2t)}$ iterations of the for loop. In each iteration, it calls the \textsc{Reduce} subroutine, multiplies $H$ by a vector $\vec{y}_i\in R^n$ of sparsity at most~$3t/2$, and adds two vectors from $R^n$. By~\cref{lem:reduce}, $\textsc{Reduce}$ performs $m\cdot \poly(\log{n})\leq t^{1+\alpha}\cdot\poly(\log{n})$ arithmetic operations over~$R$. By~\cref{prop:naive-sparse-mm} $H$ can be multiplied by an $n$-dimensional vector~$\vec{y}_i$ using $\norm{\vec{y}_i}_0\cdot\poly(\log{n})\leq t\cdot\poly(\log{n})$ operations over~$R$. After the loop, the algorithm simply sums up $O(\log{n})$ vectors from~$R^n$. Thus, the algorithm~$\cR$ performs at most $n^{1+\alpha}\cdot\poly(\log{n})$ operations over~$R$.
		
		We now argue correctness of the algorithm~$\cR$. Assume that $\norm{\vec{x}}_0\leq t$.  In this case, \cref{lem:reduce} guarantees that   for all $i\leq \ceil{\log(2t)}$,
		\begin{align*}\label{eq:reduce}
			\norm{\vec{x}-\sum_{j\leq i}\vec{y}_j}_0\leq t/ 2^i \;.
		\end{align*}
		In particular, each $\vec{y}_i$ satisfies $\norm{\vec{y}_i}_0\leq 3t/2$, and the algorithm performs all $\ceil{\log(2t)}$ iterations of the for loop. 
		Since $\norm{\vec{x}-\sum_{j\leq \log(2t)}\vec{y}_j}_0=0$, we conclude that $\vec{x}=\sum_{i=1}^{\ceil{\log(2t)}}\vec{y}_i$. This finishes the proof of the theorem.
	\end{proof}

\end{document}